\documentclass[psamsfonts,11pt]{amsart}

\usepackage[usenames,dvipsnames]{color}
\usepackage{xifthen}
\usepackage[utf8]{inputenc}
\usepackage{amsmath}
\usepackage{mathtools}
\usepackage{amsthm}
\usepackage{amsfonts}
\usepackage{amssymb}
\usepackage[margin=1in]{geometry}
\usepackage[all]{xy}
\usepackage[english]{babel}
\usepackage{graphicx}
\usepackage{xparse,zref-savepos}
\usepackage{tikz}
\usepackage{tikz-cd}
\usepackage{hyperref}
\usepackage[shortlabels]{enumitem}
\usetikzlibrary{matrix,arrows,decorations.pathmorphing,shapes.misc}
\usepackage{capt-of}
\usepackage[normalem]{ulem}
\usepackage{appendix}
\usepackage{bbm}
\numberwithin{equation}{section}
\usepackage{courier}
\usepackage{enumitem}
\usepackage[capitalize,noabbrev]{cleveref}
\usepackage{pdfsync}
\newtheorem{theorem}{Theorem}[section]
\newtheorem{lemma}[theorem]{Lemma}
\newtheorem{proposition}[theorem]{Proposition}

\newtheorem{corollary}[theorem]{Corollary}
\theoremstyle{definition}

\newtheorem{counterexample}[theorem]{Counterexample}

\newtheorem{definition}[theorem]{Definition}

\newcommand{\eps}{\varepsilon}

\newcommand{\area}{\textup{area}}

\newcommand{\R}{\mathbb{R}}
\newcommand{\N}{\mathbb{N}}

\newcommand{\pers}{\textup{pers}}

\newcommand{\cmon}[1][]{
	\ifthenelse{\isempty{#1}}{\mathbf{CMon}}{\mathbf{CMon}(#1)}
}
\newcommand{\ab}[1][]{
	\ifthenelse{\isempty{#1}}{\mathbf{Ab}}{\mathbf{Ab}(#1)}
}
\newcommand{\cmonti}[1][]{
	\ifthenelse{\isempty{#1}}{\mathbf{CMon}^{\mathbf{ti}}}{\mathbf{CMon}(#1)^{\mathbf{ti}}}
}
\newcommand{\abti}[1][]{
	\ifthenelse{\isempty{#1}}{\mathbf{Ab}^{\mathbf{ti}}}{\mathbf{Ab}(#1)^{\mathbf{ti}}}
}

\newcommand{\spt}{\textup{spt}}
\newcommand{\esssup}{\textup{ess sup}}
\newcommand{\Adm}{\textup{Adm}}
\newcommand{\Opt}{\textup{Opt}}
\newcommand{\OT}{\textup{OT}}
\newcommand{\w}{\mathbb{W}}
\newcommand{\ow}{\overline{\mathbb{W}}}

\newcommand{\supp}{\textup{supp}}

\makeatletter
\def\moverlay{\mathpalette\mov@rlay}
\def\mov@rlay#1#2{\leavevmode\vtop{%
		\baselineskip\z@skip \lineskiplimit-\maxdimen
		\ialign{\hfil$\m@th#1##$\hfil\cr#2\crcr}}}
\newcommand{\charfusion}[3][\mathord]{
	#1{\ifx#1\mathop\vphantom{#2}\fi
		\mathpalette\mov@rlay{#2\cr#3}
	}
	\ifx#1\mathop\expandafter\displaylimits\fi}
\makeatother

\newcommand{\Expect}{{\rm I\kern-.3em E}}

\makeatletter
\@ifundefined{zsaveposx}{\let\zsaveposx\zsavepos}{}
\newcounter{hposcnt}
\renewcommand*{\thehposcnt}{hpos\number\value{hposcnt}}
\NewDocumentCommand{\lplabel}{o m}{%
	\stepcounter{hposcnt}%
	\zsaveposx{\thehposcnt l}%
	\zref@refused{\thehposcnt l}%
	\zref@refused{hpos0l}%
	\makebox[0pt][r]{\makebox[\dimexpr\zposx{\thehposcnt l}sp-\zposx{hpos0l}sp][l]{#2}}%
	\IfNoValueF{#1}
	{\def\@currentlabel{#2}\ltx@label{#1}}
}
\makeatother
\AtBeginDocument{\zsaveposx{hpos0l}}

\newlist{thmenum}{enumerate}{1}
\setlist[thmenum, 1]{label=(\arabic*), ref=\thetheorem ~(\arabic*)}

\title{Learning on Persistence Diagrams as Radon Measures} 

\author[A. Elchesen]{Alex Elchesen}
\address{Department of Mathematics, Colorado State University, USA}

\author[I. Hartsock]{Iryna Hartsock}
\address{Department of Mathematics, University of Florida, USA}

\author[J. A. Perea]{Jose A. Perea}
\address{Department of Mathematics and Khoury College of Computer Sciences, Northeastern University, USA}

\author[T. Rask]{Tatum Rask}
\address{Department of Mathematics, Colorado State University, USA}

\begin{document}

\maketitle

\begin{abstract}
Persistence diagrams are common descriptors of the topological structure of data appearing in various classification and regression tasks. They can be generalized to Radon measures supported on the birth-death plane and endowed with an optimal transport distance. Examples of such measures are expectations of probability distributions on the space of persistence diagrams. In this paper, we develop methods for approximating continuous functions on the space of  Radon measures supported on the birth-death plane, as well as their utilization in supervised learning tasks. Indeed,  we show that any continuous function defined on a compact subset of the space of such measures (e.g., a classifier or regressor) can be approximated arbitrarily well by polynomial combinations of features computed using a continuous compactly supported function on the birth-death plane (a template). We provide insights into the structure of relatively compact subsets of the space of Radon measures, and test our approximation methodology  on various data sets and supervised learning tasks.
\end{abstract}

\section{Introduction}\label{sec:intro}

Persistent homology is a popular tool in computational topology used to extract meaningful topological features from a dataset. The input to a persistent homology computation is a nested sequence $K_0\subset K_1\subset \dots \subset K_n$ of topological spaces (usually simplicial complexes). 
For example, given a dataset $X\subset \R^d$, computing the  \emph{Vietoris-Rips}  or  \emph{\^{C}ech} filtration of $X$ yields a nested family of simplicial complexes encoding the topology of $X$ across different scales. The output of a persistent homology computation is a \emph{persistence diagram}, a multiset of points supported on the birth-death plane $\w = \{(b,d) \in \R^2 \ | \ b < d\}$. Each point $(b,d)$ of a persistence diagram represents the lifetime of some topological feature across the filtration, with $b$ being the \emph{birth time} of a feature and $d$ being the \emph{death time}. When $X$ is a point cloud and we apply persistent homology to one of the filtrations described above, a point $(b,d)$ in the persistence diagram represents the scales over which a topological feature in $X$ persists.

While many recent algorithmic advances have greatly improved the speed of persistent homology computations, computing persistence can still be intractable for very large datasets. For this reason, it is common to compute persistence on a random subsample of the original dataset. Repeating this process results in an i.i.d.\ sample $D_1,\dots, D_n$ of persistence diagrams, each computed using a different random subsample of the data. 
One hopes that an ``expected persistence diagram'' would exist, and that it could serve as a proxy for the persistence diagram of the whole dataset. Chazal and Divol \cite{chazal2019density} show that such an expectation exists, but that it is a measure supported on $\w$ rather than a persistence diagram. Thus, the expected persistence \emph{measure} can be viewed as a topological descriptor of the original dataset.

More generally, consider a \emph{metric measure space}, i.e., a metric space $(X,d)$ equipped with a Borel probability measure $\mu$. Applying persistence directly to $X$ results in a topological descriptor of $X$ with no information about the measure $\mu$.
To overcome this problem, we can randomly sample $X$ with respect to $\mu$. Applying persistence to the sampled points clouds, we obtain i.i.d.\ persistence diagramds $D_1,\dots,D_n$, whose expected persistence measure can be estimated using \emph{kernel density estimation}. This pipeline is illustrated in Figure \ref{fig:pipeline}. The result is an expected persistence measure which encodes information about both the measure $\mu$ and the topology of $X$. Thus, the persistence measure is better suited for classification tasks involving metric measure spaces, and we explore its usage here.

In this paper, we propose treating persistence measures as topological descriptors for geometric objects. Our goal is to then use these descriptors as features in supervised learning tasks as follows. Let $\mathcal{M}$ be the space of \emph{Radon measures} on $\w$ equipped with the \emph{relative $\infty$-optimal transport distance} $\OT_\infty$. Fix  a continuous function $F:\mathcal{M}\to \R$, and a finite training sample $\mu_1,\dots, \mu_n\in \mathcal{M}$ together with the values $F(\mu_1),\dots, F(\mu_n)$. The goal in supervised learning is  to approximate the (e.g., classifier or regressor) $F$ using just the observed values $F(\mu_1),\dots, F(\mu_n)$. In order to construct such approximations, it is useful to find dense subsets of ``simpler'' functions whose form are amenable to optimization. A classical example of this is the Weierstrass approximation theorem which states that any continuous function defined on a closed interval $[a,b]$ can be uniformly approximated arbitrarily well by a polynomial function. The coefficients of an approximating polynomial can be computed by minimizing the sum of the squared errors between the polynomial and the ground truth evaluated on the training sample.

\begin{figure}[t]
	\begin{center}
		\includegraphics[trim={0cm 0.5cm 0cm 0cm},clip,width=1\linewidth]{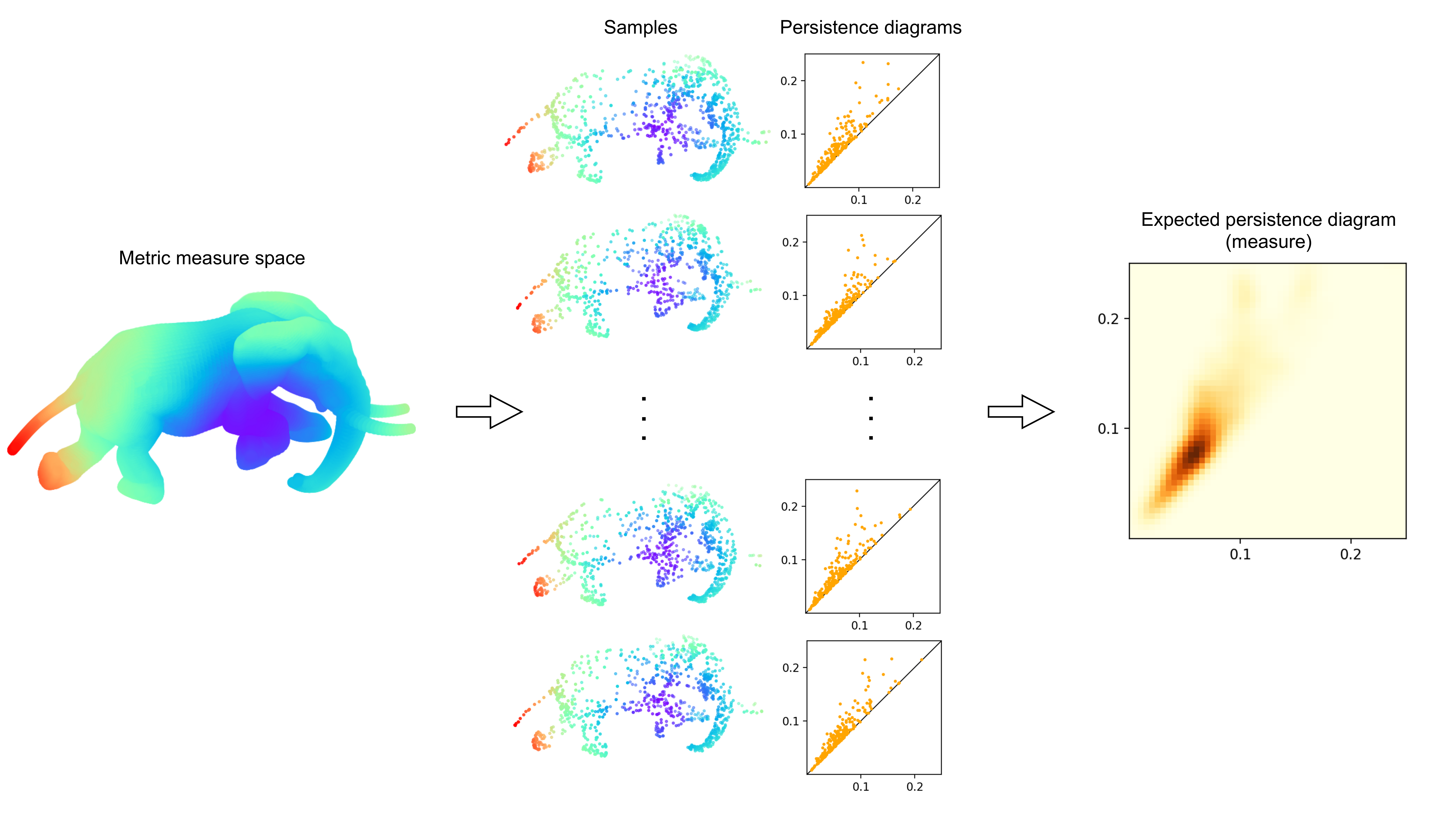}
		\caption{Our computational pipeline. From a metric measure space (color being density) we sample point clouds  and  compute their persistence diagrams. Next, we use the obtained persistence diagrams together with kernel density estimation to approximate the expected persistence measure. }\label{fig:pipeline}
	\end{center}
\end{figure}


\subsection{Contributions}
 Our first goal is to find easy-to-describe, compact-open dense subsets of the space of continuous functions on $(\mathcal{M}, \OT_\infty) $. Secondly, since our approach to approximating  a given   $F:\mathcal{M}\to \R$ will have theoretical guarantees only on compact sets, 
 it is important that we clarify the structure of said subsets so that practitioners can evaluate the application-specific suitability of our method.
 Lastly, we would like to instantiate the theory in a computational scheme for approximating continuous functions on measure spaces.

We achieve the first goal by showing that any continuous functions on $\mathcal{M}$ can be approximated arbitrarily well on compact sets by polynomial combinations of \emph{linear representations} (Theorem \ref{thm:FirstApproximation}). A linear representation is a real-valued function on $\mathcal{M}$ of the form $\mu\mapsto \int fd\mu$ for some continuous and compactly supported $f:\w\to \R$. We then strengthen this further to show that it suffices to fix a single function $f$, called a \emph{template function}, and consider the collection of all  its scalings and translations $f(a\mathbf{x} + \mathbf{b})$. We show that polynomial combinations of linear representations induced by this collection can be used to approximate continuous functions on compact sets as well (Theorem \ref{thm:function_approx}).

We then study the topology of  $(\mathcal{M},\OT_\infty)$ and its relevant compact subsets. 
Indeed, we first show that the subspace of \emph{off-diagonally finite} measures $\mathcal{M}_f^\infty$ (Definition \ref{def:off_d_finite}) is the completion of the space of finite measures (Theorem \ref{thm:completion_finite}), which are the ones expected to appear in practice. We then give necessary conditions for the relative compactness of subsets of $\mathcal{M}_f^\infty$. A complete characterization of compactness in $\mathcal{M}$  is still an open problem.

Lastly, we develop a computational pipeline for approximating continuous functions on measures space, and demonstrate its utility in classification tasks on four different datasets.

\subsection{Related Work}
The idea of using template functions in TDA to approximate continuous functions originated in \cite{perea2022approximating}. There, the authors prove density results for the space of continuous functions on persistence diagrams and give a characterization of the relatively compact subsets of diagram space. A complete characterization of relative compact subset of generalized persistence diagrams on metric pairs is given in \cite{bubenik2022topological}. In \cite{polanco2019adaptive}, the authors explore adaptive methods for feature selection using template functions. Expected persistence measure are introduced in \cite{chazal2019density}, where the authors give necessary conditions for these expectations to have densities with respect to Lebesgue measure. In \cite{divol2019understanding}, the authors study an extension of the bottleneck distance to the spaces of Radon measures using the theory of partial optimal transport introduced in \cite{figalli2010new}. The persistence image, introduced in \cite{adams2017persistenceimages}, is a particular example of the kernel density estimate of the expected persistence measure.

\section{Background}\label{sec:background}

\subsection{Function Spaces} The space of continuous functions between topological spaces $X,Y$ will be denoted $C(X,Y)$. When $Y = \R$ with its usual topology, we will sometimes denote $C(X,\R)$ more simply by $C(X)$. The subspace of $C(X)$ consisting of the compactly supported, real-valued, continuous functions will be denoted by $C_c(X)$.

There are several different topologies that we will use on functions spaces throughout this paper. Let $(K_n)$ be an increasing sequence of compact subsets of $\w$ such that $\bigcup_n K_n  = \w$. Then we have inclusions $C_c(K_n)\hookrightarrow C_c(\w)$ for all $n$. The \emph{strict inductive limit topology} on $C_c(\w)$ is the finest locally convex topology for which all of the inclusions $C_c(K_n)\hookrightarrow C_c(\w)$ are continuous. This topology is independent of the choice of the compact sets $K_n$. 

We will also make use of the \emph{compact-open topology} on functions spaces of the form $C(X)$. In the case of real-values functions on $X$, this is the topology generated by sets of the form $V(K,U) := \{f\in C(X) \ | \ f(K)\subset U\}$ for $K\subset X$ compact and $U\subset \R$ open.

\subsection{Measure Theory}\label{sec:measure_theory}

Throughout, $X = (X,d)$ denotes a locally compact metric space.

\begin{definition}\label{def:radon}
	A Borel measure $\mu$ on $X$ is a \emph{Radon measure} if it satisfies
	\begin{enumerate}
		\item For all open sets $U\subset X$, $\mu(U) = \sup\{\mu(K) \ | \ K\subset U \textup{ compact}\}$ (inner regularity).
		\item For all Borel sets $E\subset X$, $\mu(E) = \inf\{\mu(U) \ | \ U\subset E \textup{ open}\}$ (outer regularity).
		\item $\mu(K) <\infty$ for all $K\subset X$ compact (local finiteness).
	\end{enumerate}
	The set of Radon measures on a space $X$ will be denoted $\mathcal{M}(X)$.
\end{definition}

Throughout this section, we will provide useful information and results  about measures that will be used throughout. 
\begin{definition}
	Let $\mu$ be a Borel measure on $X$ and let $f:X\to \R$ be a measurable function. The \emph{essential supremum of $f$ (with respect to $\mu$)} is defined by
	\[\esssup_{\mu}(f) := \inf\{a\in \R \ | \ \mu(f^{-1}(a,\infty)) = 0)\}.\]
\end{definition}
In other words, the essential supremum of $f$ is the infimal $a\in \R$ for which $f(x) \leq a$ for $\mu$-almost every $x\in X$.

\begin{definition}
	Let $\mu$ be a Borel measure on $X$. The \emph{support of $\mu$} is the set
	\[ \spt(\mu) := \{x\in X \ | \ \mu(N_x) > 0 \textup{ for every neighborhood $N_x$ of $x$}\}.\]
\end{definition}
Note that $\spt(\mu)$ is closed. Indeed, if $x\in X\backslash \spt(\mu)$ and $U_x$ is an open neighborhood of $x$ with $\mu(U_x) = 0$ then $U_x \subset X\backslash \spt(\mu)$ so that $X\backslash \spt(\mu) = \bigcup_{x\in X\backslash \spt(\mu)} U_x$ is open.
\begin{lemma}\label{lem:cplmt_spt_null}
	Let $\mu$ be an inner regular Borel measure on $X$. If $A\subset X\backslash \spt(\mu)$ is measurable then $\mu(A) = 0$.
	\begin{proof}
		It suffices to show that $\mu(X\backslash \spt(\mu)) = 0$. Let $K\subset X\backslash \spt(\mu)$ compact. For each $x\in K$, we can find an open set $U_x$ containing $x$ with $\mu(U_x) = 0$. Since $K$ is compact, we can cover $K$ with finitely many such $U_x$ and hence $\mu(K) = 0$. Since $X\backslash \spt(\mu)$ is open, by inner regularity we have $\mu(X\backslash \spt(\mu)) = 0$.
	\end{proof}
\end{lemma}

Below, under certain conditions, we relate the essential supremum of a measurable function and the support of a Borel measure.
\begin{proposition}
	Let $\mu$ be an inner regular Borel measure on $X$ and let $f:X\to\R$ be measurable. Then
	\[\esssup_\mu(f) = \sup\{f(x) \ | \ x\in \spt(\mu)\}.\]
	\begin{proof}
		Let $\alpha := \esssup_\mu(f)$ and $\beta:=  \sup\{f(x) \ | \ x\in \spt(\mu)\}$. First, we show that $\alpha \geq \beta$. Let $x\in \spt(\mu)$. Then for all $\eps > 0$, we have $x\in f^{-1}(f(x)-\eps, \infty)$. Since $x\in \spt(\mu)$, we have $\mu(f^{-1}(f(x)-\eps, \infty)) > 0$. It follows from the definition of $\alpha$ that $f(x) - \eps \leq \alpha$. Since $\eps > 0$ was arbitrary, we have $f(x) \leq \alpha$ for all $x\in \spt(\mu)$. Then $\alpha \geq \beta$ by definition of $\beta$.
		
		Next, we show that $\alpha \leq \beta$. Note that if $x\in f^{-1}(\beta,\infty)$ then $x\not\in \spt(\mu)$. Hence $f^{-1}(\beta,\infty)\subset X\backslash \spt(\mu)$. By Lemma \ref{lem:cplmt_spt_null}, we have $\mu(f^{-1}(\beta,\infty)) = 0$. Hence $\alpha \leq \beta$ by definition of $\alpha$.
	\end{proof}
\end{proposition}

We will use the following notion of convergence throughout.

\begin{definition}
A sequence $(\mu_n)$ in $\mathcal{M}(X)$ is said to converge \emph{vaguely} to $\mu$ in $\mathcal{M}(X)$, denoted $\mu_n\xrightarrow{v} \mu$, if $\int f d\mu_n\to \int fd\mu$ for any $f \in C_c(\mathbb{W})$. 
\end{definition}

\subsection{Partial Optimal Transport} In this section, we define partial optimal transport and the corresponding $\infty$-partial optimal transport distance. Partial optimal transport was originally introduced in \cite{figalli2010new} and studied in \cite{divol2019understanding} in the context of persistence.

For simplicity, we denote $\mathcal{M}(\mathbb{W})$, the set of Radon measure on $\mathbb{W}$, by $\mathcal{M}$.
Fix $1\leq q\leq \infty$ and let $d$ be the pseudometric on $\ow$ given by $d(\mathbf{x},\mathbf{y}) := \min(\|\mathbf{x}-\mathbf{y}\|_q, \|\mathbf{x}-\Delta\|_q + \|\mathbf{y}-\Delta\|_q)$, where $\|\mathbf{x} - \Delta\|_q = \inf_{\mathbf{z}\in \Delta} \|\mathbf{x}-\mathbf{z}\|_q$. Note that $d(\mathbf{x},\Delta) = \|\mathbf{x}-\Delta\|_q$ and, if $\mathbf{y}\in \Delta$, then $d(\mathbf{x},\mathbf{y}) = \|\mathbf{x}-\Delta\|_q$.
\begin{definition}
	Given a Radon measure $\mu\in \mathcal{M}$, define the \emph{$\infty$-persistence} of $\mu$ by
	\[ \pers_\infty(\mu) := \sup\{d(\mathbf{x},\Delta) \ | \ \mathbf{x}\in \spt(\mu)\} = \sup\{\|\mathbf{x}- \Delta\|_q \ | \ \mathbf{x}\in \spt(\mu)\}.\]
	and define
	\[ \mathcal{M}^\infty := \{\mu \in \mathcal{M} \ | \ \pers_\infty(\mu) < \infty\}.\]
\end{definition}
%

\begin{definition}
	Let $\mu,\nu\in \mathcal{M}^\infty$. A \emph{coupling} (or \emph{admissible transport plan}) between $\mu$ and $\nu$ is a Radon measure $\pi\in \mathcal{M}(\ow \times \ow)$ such that, for all Borel sets $A,B\in \w$,
	\[\pi(A\times \ow) = \mu(A) \quad \textup{and} \quad \pi(\ow\times B) = \nu(B).\]
	The set of all couplings between $\mu$ and $\nu$ is denoted $\Adm(\mu,\nu)$. The \emph{$\infty$-cost} of $\pi \in \Adm(\mu,\nu)$ is defined by
	\[ C_\infty(\pi) := \esssup_\pi(d) = \sup\{d(\mathbf{x},\mathbf{y}) \ | \ (\mathbf{x},\mathbf{y})\in \spt(\pi)\}.\]The \emph{partial $\infty$-optimal transport distance} is defined by
	\[ \OT_\infty(\mu,\nu) := \inf_{\pi\in \Adm(\mu,\nu)}C_\infty(\pi).\]
\end{definition}


We will use the following facts throughout about $\OT_\infty$ throughout.
\begin{theorem}[{\cite[Proposition 3.10]{divol2019understanding}}]\label{thm:completeness}
	The space $(\mathcal{M}^\infty, \textup{OT}_\infty)$ is complete.
\end{theorem}

\begin{proposition}[{\cite[Proposition 3.11]{divol2019understanding}}]\label{prop:OT_conv_implies_vague_conv_and_pers_conv}
	Let $\mu, \mu_1, \mu_2, \ldots \in \mathcal{M}^\infty$. If $\OT_{\infty}(\mu_n, \mu) \to 0$, then $(\mu_n) \xrightarrow{v} \mu$ and $\pers_\infty(\mu_n) \to \pers_\infty(\mu)$.
\end{proposition}

\subsection{Expected Persistence Measures} Let $\mathbb{X}$ be a random point cloud. For example, $\mathbb{X}$ might be obtained by randomly sampling a fixed number of points from a manifold. We can compute the persistence diagram of a filtration $\mathcal{K}(\mathbb{X},r)$, where $\mathcal{K}(\mathbb{X},r)$ might be either the Vietoris-Rips or the \v{C}ech filtration. The expectation of this random persistence diagram will not itself be a persistence diagram, but  a Radon measure supported on $\w$. In \cite{chazal2019density}, Chazal and Divol prove that, under certain conditions, the expected persistence measure has a density with respect to Lebesgue measure.

\begin{theorem}[{\cite[Theorem 3.3]{chazal2019density}}]\label{thm:chazal_density} Fix $n\geq 1$. Assume that $M$ is a real analytic compact $d$-dimensional connected submanifold and that $X$ is a random variable on $M^n$ having a density with respect to the Hausdorff measure. Then for $s\geq 1$, the expected persistence measure $\mathbb{E}[D_s(\mathcal{K}(\mathbb{X}))]$ has a density with respect to the Lebesgue measure on $\w$. Moreover, $\mathbb{E}[D_s(\mathcal{K}(\mathbb{X}))]$ has a density with respect to Lebesgue measure on the vertical line $\{0\} \times [0,\infty)$.
\end{theorem}

We will utilize the fact that expected persistence measures have a density in Section \ref{sec:KDE} and in our experiments, where we apply kernel density estimation to approximate an expected persistence measure from a sample of persistence diagrams.

\section{Learning Continuous Functions on Measure Spaces} In this section, we give theoretical justification for our learning scheme. We use a version of the Stone-Weierstrass theorem to show that continuous functions on $C(\mathcal{M},\mathbb{R})$ can be approximated arbitrarily well by polynomial combinations of features computed using a \emph{template function} --- a compactly supported continuous function on $\w$.

\subsection{The Stone-Weierstrass Theorem for the Compact-open Topology} The main tool used throughout this section is the following  version of the Stone-Weierstrass theorem.

\begin{theorem}[\cite{kelley1955generaltopology}] Let $X$ be a metric space. If $A$ is a subalgebra of $C(X)$ that separates points in $X$ and contains the constant functions, then $A$ is dense in $C(X)$ with respect to the compact-open topology.
\end{theorem}

If $S$ is any subset of $C(X)$, then the \emph{subalgebra generated by $S$}, denoted $\langle S \rangle$, is the subalgebra of $C(X)$ consisting of all finite $\R$-linear combinations of finite products of elements of $S$. That is, $\langle S \rangle$ consists of all elements of $C(X)$ of the form $p(f_1,\dots,f_N)$, where $N\in \mathbb{N}$, $f_1,\dots,f_N\in S$, and $p\in \R[x_1,\dots,x_N]$ is a polynomial in $N$ variables with real coefficients.

For any $S\subset C(X)$, $\langle S \rangle$ automatically contains the constant functions. Moreover, if $S$ separates points in $X$ then so does  $\langle S \rangle$. Thus, by the Stone-Weierstrass theorem, if $S$ separates points then $\langle S \rangle$ is dense is $C(X)$ with respect to the compact-open topology.

\subsection{Compact-open Dense Subsets of $C(\mathcal{M},\R)$} 

A very natural way of constructing a continuous real-valued function on $\mathcal{M}$ is to start with a (compactly supported, continuous) function $f:\mathbb{W}\to \R$ and then to define $\hat{f}:\mathcal{M}\to \R$ according to $\mu\mapsto \int_\mathbb{W} fd\mu$. Continuity of $\hat{f}$ follows from Proposition \ref{prop:OT_conv_implies_vague_conv_and_pers_conv}. In this section, we show first that the algebra generated by such functions is dense in $C(\mathcal{M},\R)$. We then strengthen this to show that it suffices to start with a single function $f\in C_c(\w)$ and consider the subalgebra generated by all scaling and translations of $f$. The function $f$ is what we refer to as a template function.


\begin{definition}
	Given $f\in C_c(\mathbb{W})$, let $\hat{f}:\mathcal{M}\to \R$ be given by $\mu \mapsto \int f d\mu$. Given a subset $A\subset C_c(\mathbb{W})$, we write $\hat{A} := \{\hat{f} \ | \ f\in A\}$. In particular, we have $\widehat{C}_c(\mathbb{W}) = \{\hat{f} \ | \ f\in C_c(\mathbb{W})\}$.
\end{definition}

Next, we observe that $\widehat{C}_c(\mathbb{W}) = \{\hat{f} \ | \ f\in C_c(\mathbb{W})\}$ separates points on $\mathcal{M}$. This fact is essentially the content of the well-known Riesz representation theorem.

\begin{theorem}[The Riesz representation theorem]\label{thm:Riesz} Let $X$ be a locally compact Hausdorff space. Given a positive linear functional $I:C_c(X)\to \R$, there exists a unique positive Radon measure $\mu\in \mathcal{M}^+(X)$ such that $I(f) = \int_X fd\mu$ for all $f\in C_c(X)$.
\end{theorem}

An immediate corollary of the Riesz representation theorem is that for two measures $\mu,\nu\in \mathcal{M}$, we have $\mu = \nu$ if and only if $\int fd\mu = \int f d\nu$ for all $f\in C_c(\mathbb{W})$. Equivalently, $\mu$ and $\nu$ are distinct if and only if there exists some $f\in C_c(\mathbb{W})$ for which $\int fd\mu \neq \int f d\nu$. Said differently, we have the following.


\begin{corollary}\label{cor:meas_determined_by_cmptly_supp_fnts}
	$\widehat{C}_c(\mathbb{W})$ seperates points in $\mathcal{M}$, i.e., for every pair of distinct $\mu,\nu\in \mathcal{M}$, there exists $f\in C_c(\mathbb{W})$ such that $\hat{f}(\mu)\neq \hat{f}(\nu)$.
\end{corollary}

Combining the preceding corollary with the Stone-Weierstrass theorem, we have the following.

\begin{theorem}\label{thm:FirstApproximation}
	The subalgebra generated by $\hat{C}_c(\mathbb{W})$ is dense in $C(\mathcal{M},\R)$ in the compact-open topology. More concretely, given a function $F:\mathcal{M}\to\R$, a compact set $K\subset \mathcal{M}$, and $\eps > 0$, there exists an $N\in \N$, functions $f_1,\dots,f_N\in C_c(\mathbb{W})$, and a polynomial $p\in \R[x_1,\dots,x_N]$ such that
	\begin{equation}\label{eq:approx}  \sup_{\mu\in K} \left|p\left(\int f_1d\mu,\dots,\int f_Nd\mu\right) - F(\mu)\right|<\eps.
	\end{equation}
\end{theorem}

The preceding theorem is somewhat unsatisfactory from a computational point of view. It may be difficult or even impossible to encode an arbitrary compactly supported continuous function on a computer. Moreover, numerical issues may make it difficult to compute the integrals $\int f_id\mu$ appearing in \eqref{eq:approx}.

In what follows, we will show that it suffices to choose a single $f\in C_c(\w)$ and consider the collection consisting of all translations and scalings of $f$. The idea is to show that subalgebras of $C_c(\w)$ that are dense with respect to the strict inductive limit topology give rise to subalgebras of $C(\mathcal{M},\R)$ that are dense with respect to the compact-open topology. We then show that, given $f\in C_c(\w)$, the subalgebra generated by the collection of all translations and scalings of $f$ is a dense subalgebra of $C_c(\w)$.

\begin{proposition} \label{prop:dense_algebra_of_CcW_separates_measures}
	Let $\mathcal{A}$ be a subalgebra of $C_c(\mathbb{W})$, dense with respect to the strict inductive limit topology, and let $\mu,\nu\in \mathcal{M}$. Then there exists $g\in\mathcal{A}$ such that $\int g d\mu \neq \int g d\nu$.
		\begin{proof}
		Let $\mu, \nu \in \mathcal{M}$ and $\mu \neq \nu$. By Corollary \ref{cor:meas_determined_by_cmptly_supp_fnts}, there is some $\tilde{f} \in C_c(\mathbb{W})$ such that $\int \tilde{f}d\mu \neq \int \tilde{f} d\nu$. Let $a := | \int fd\mu - \int f d\nu | > 0$. Let $\mathbb{W}=  \bigcup_{n=1}^{\infty} K_n$, where every $K_n \subset \mathbb{W}$ is compact and $K_n \subset K_{n+1}$ for all $n$. Without loss of generality, we can assume that $\supp(\tilde{f}) \subset K_m$ for some $m$.  Let $F, G: C_c(\mathbb{W}) \to \R$ be the linear maps given by $F(f) = \int f d \mu$ and $G(f) = \int f d \nu$ for all $f \in C_c(\mathbb{W})$, respectively. In the strict inductive limit topology on $C_c(\mathbb{W})$, the maps $F$ and $G$ are continuous if and only if their restrictions to every $C_c(K_n) =\{f \in C_c(\mathbb{W}) \ | \ \supp (f) \subset K_n\}$, endowed with the $\|\cdot\|_\infty$ norm, is continuous.
		Since the restriction of $F$ and $G$ to $C_c(K_m)$ are continuous at $\tilde{f}$, there exists $\delta>0$ and $g \in B_\delta(\tilde{f}) \cap \mathcal{A}$ such that $|F(\tilde{f})-F(g) | < \frac{a}{2}$ and $|G(\tilde{f})-G(g)| < \frac{a}{2}$.
		By the triangle inequality,
		\[|F(\tilde{f})-G(\tilde{f}) | \leq |F(\tilde{f})-F(g) | +|F(g)-G(g)|+ |G(g)-G(\tilde{f})|.\] Therefore,
		\[|F(g)-G(g)| \geq |F(\tilde{f})-G(\tilde{f}) | - |F(\tilde{f})-F(g) | - |G(g)-G(\tilde{f})| > a-\frac{a}{2}-\frac{a}{2}=0.\]  Hence $\int g d\mu = F(g) \neq G(g) = \int g d\nu$.
	\end{proof}
\end{proposition}

\begin{definition}
	Given $f\in C_c(\mathbb{W})$, define
	\[ T(f) := \{g\in C_c(\mathbb{W}) \ | \ g(\mathbf{x}) = f(a\mathbf{x} + \mathbf{b})\textup{ for some } a\in \R,\, \mathbf{b}\in \R^2\},\]
	and define $\hat{T}(f) := \{\hat{g} \ | \ g\in T(f)\}$. We refer to $f$ as a template function and to $T(f)$ as a template system.
\end{definition}

\begin{theorem}\label{thm:transformations_dense}
	Fix $f\in C_c(\mathbb{W})$ and let $\mathcal{A}$ be the subalgebra of $C_c(\mathbb{W})$ generated by $T(f)$. Then $\mathcal{A}$ is dense in $C_c(\mathbb{W})$ with respect to the strict inductive limit topology.
		\begin{proof}  Let $K_n$ be an increasing sequence of compact sets with $\bigcup_n K_n = \w$. We can assume without loss of generality that $\supp(g)\subset K_1$. We claim that $T(f)\cap C_c(K_n)$ separates points on $K_n$. Indeed, given distinct $x,y\in K_n$, we can scale and translate $f$ to obtain a function $\tilde{f}\in T(f)$ whose support is contained in $K_n$ and for which $\tilde{f}(x) \neq 0$ but $\tilde{f}(y) = 0$. It then follows from the Stone-Weierstrass theorem for compact spaces that $\langle T(f)\cap C_c(K_n)\rangle$ is dense in $C_c(K_n) = C(K_n)$ with respect to the sup-norm.
		
		Now let $g\in C_c(\w)$ and let $U$ be an open neighborhood of $g$. By definition of the strict inductive limit topology, the inclusions $C_c(K_n)\hookrightarrow C_c(\w)$ are all continuous. Hence $U\cap C_c(K_n)$ is open in $C_c(K_n)$. By density of $\langle T(f)\cap C_c(K_n)\rangle$ in $C_c(K_n)$, there is some $\alpha\in \langle T(f)\cap C_c(K_n)\rangle$ with $\alpha\in U$. Since $\langle T(f)\cap C_c(K_n)\rangle\subset \langle T(f)\rangle$, the result follows.
	\end{proof}
\end{theorem}

\begin{theorem}\label{thm:function_approx}
	For a fixed $f\in C_c(\w)$, the subalgebra generated by $\hat{T}(f)$ is dense in $C(\mathcal{M},\R)$ in the compact-open topology. More concretely, given a function $F:\mathcal{M}\to\R$, a compact set $K\subset \mathcal{M}$, and $\eps > 0$, there exists an $N\in \N$, scalars $a_1,\dots, a_N\in \R$, vectors $\mathbf{b}_1,\dots,\mathbf{b}_N\in \R^2$, and a polynomial $p\in \R[x_1,\dots,x_N]$ such that
	\begin{equation}\label{eq:approx2}  \sup_{\mu\in K} \left|p\left(\int f(a_1\mathbf{x} + \mathbf{b}_1)d\mu(\mathbf{x}),\dots,\int f(a_N\mathbf{x} + \mathbf{b}_N)d\mu(\mathbf{x})\right) - F(\mu)\right|<\eps.
	\end{equation}
	\begin{proof}
		By Theorem \ref{thm:transformations_dense}, the subalgebra generated by $T(f)$ is dense in $C_c(\w)$ with respect to the strict inductive limit topology. Then by Proposition \ref{prop:dense_algebra_of_CcW_separates_measures}, $T(f)$ separates points in $\mathcal{M}$. It then follows from the Stone-Weierstrass theorem that the subalgebra generated by $\hat{T}(f)$ is dense in $C(\mathcal{M},\R)$ with respect to the compact-open topology.
	\end{proof}
\end{theorem}

\section{Kernel Density Estimation of the Expected Persistence Measure}\label{sec:KDE} In the preceding section, we showed that for a continuous functions $F:\mathcal{M}\to \R$ and $\mu\in \mathcal{M}$, $F(\mu)$ can be approximated by polynomial combinations of features of the form $\int f d\mu$. We will view the measure $\mu$ as being the expected persistence measure associated to some random process for generating persistence diagrams. This is the case, for example, under the conditions of Theorem \ref{thm:chazal_density}. In this setting, we do not have access to the true expected persistence measure $\mu$, but rather to a sample $D_1,\dots,D_n$ of i.i.d.\ random persistence diagrams with $\mu = \mathbb{E}(D_1)$. In order to approximate the density of $\mu$, we can apply some method of density estimation. We will focus on kernel density estimation.


\begin{definition} \label{def:kernel}
	A \emph{kernel} (on $\R^2$) is a non-negative, Lebesgue integrable function $K:\R^2\to \R$ such that
	\begin{enumerate}
		\item $\int_{\R^2} K(\mathbf{x})d\mathbf{x}= 1$ (normalization),
		\item $K(-\mathbf{x}) = K(\mathbf{x})$ for all $\mathbf{x}\in \R^d$ (symmetry).
	\end{enumerate}
\end{definition}

There are many popular kernels such as the \emph{Gaussian} and \emph{Epanechnikov} kernels. For our experiments, we use the \emph{step} (or \emph{uniform}) kernel. For a measurable subset $A\subset \mathbb{R}^2$ satisfying $-A = A$, define the \emph{step kernel on $A$} by $K_A := \frac{1}{\area(A)}1_A$, where $1_A$ denotes the indicator function on $A$.

\begin{definition}
	Let $D_1,\dots, D_n$ be i.i.d. random persistence diagrams and suppose that the expected persistence measure $\mu = \mathbb{E}[D_1]$ has a density $\rho$ with respect to Lebesgue measure. For a fixed kernel $K$, the \emph{kernel density estimate} of $\rho$ is given by
	\begin{equation}\label{eq:KDE} \hat{\rho}(\mathbf{x}) = \frac{1}{n}\sum_{i = 1}^n \sum_{\mathbf{r}\in D_i}K(\mathbf{x}-\mathbf{r}).
	\end{equation}
	
\end{definition}

Now given a compactly supported function $f:\mathbb{W}\to \R$, we wish to approximate the integral $\int_\mathbb{W}fd\mu = \int_\mathbb{W}f(\mathbf{x})\rho(\mathbf{x}) d\mathbf{x}$. Substituting the density estimate \eqref{eq:KDE} for $\rho$, we obtain

\begin{align*}\label{eq:integral_f_d_rho}
	\int f(\mathbf{x})\hat{\rho}(\mathbf{x})d\mathbf{x} & = \int f(\mathbf{x})  \frac{1}{n}\sum_{i = 1}^n \sum_{\mathbf{r}\in D_1}K(\mathbf{x}-\mathbf{r})d\mathbf{x}\\
	& = \frac{1}{n} \sum_{i = 1}^n \sum_{\mathbf{r}\in D_i}\int f(\mathbf{x})K(\mathbf{x}-\mathbf{r})d\mathbf{x} = \frac{1}{n}\sum_{i = 1}^n \sum_{\mathbf{r}\in D_i} f*K(\mathbf{r}),
\end{align*}
where $f*K$ denotes the convolution of $f$ with $K$. Thus, for a fixed kernel $K$, the approximation of $\int fd\mu$ comes down to evaluating the convolution $f*K$ at the points of the diagrams $D_1,\dots,D_n$.

In the case that $K = K_A$ is a step kernel, the convolution $f*K$ is given by
\begin{equation}\label{eq:step_kernel_conv}
	(f*K)(\mathbf{x}) = \int f(\mathbf{u})K(\mathbf{x}-\mathbf{u})d\mathbf{u} = \int f(\mathbf{u})1_A(\mathbf{x}-\mathbf{u}) = \int_{\mathbf{x}-A}f(\mathbf{u})d\mathbf{u},
\end{equation}
where $\mathbf{x}-A = \{\mathbf{x}-\mathbf{a} \ | \ \mathbf{a}\in A\}$.

In our experiments, we make a further simplification by using step functions for our templates as well. While step functions are not continuous, we may think of a step function on a set $B$ as approximating a continuous function that takes value $1$ on $B$ and then quickly decreases to $0$. If $f = 1_B$ for some $B\subset \w$ and $K = K_A$ is a step kernel, then \eqref{eq:step_kernel_conv} reduces to $(f*K)(\mathbf{x}) = \frac{1}{\area(A)}\int_{(\mathbf{x}-A)\cap B}1d\mathbf{u} = \frac{\area{((\mathbf{x}-A)\cap B})}{\area{A}}$. If, further, $A$ and $B$ are rectangles, then the area of $(\mathbf{x}-A)\cap B$ has a simple closed-form expression.

\section{Compact Subsets of Measure Space}
Theorem \ref{thm:function_approx} implies that we can approximate continuous functions arbitrarily well on compact subsets of $\mathcal{M}$. For this reason, we would like to better understand the (relatively) compact sets of $\mathcal{M}$. It turns out that the following subspace is slightly more convenient to work with.

\begin{definition}\label{def:off_d_finite}
	Let $\mathcal{M}^\infty_f := \{\mu\in \mathcal{M}^\infty \ | \ \mu(\mathbb{W}_\eps) <\infty \textup{ for all $\eps > 0$}\}$,
	where for $\eps > 0$,  $\mathbb{W}_\eps := \{(x,y) \in \R^2 \ | \ y-x > \eps\}\subset \mathbb{W}$. Elements of $\mathcal{M}^\infty_f$ are said to be \emph{off-diagonally finite}.
\end{definition}

In this section, we study the topology of $\mathcal{M}_f^\infty$. Wee give necessary conditions for compactness in $\mathcal{M}_f^\infty$. Recall that, for a topological space $X$, a set $S  \subset X$ is relatively compact if and only if $\overline{S}$ is compact in $X$. If $X = (X,d)$ is complete metric space, then $S\subset X$ is relatively compact if and only if $S$ is totally bounded. For this reason, we will begin by showing that $(\mathcal{M}_f^\infty,\OT_\infty)$ is complete.

\subsection{Completeness of $\mathcal{M}_f^\infty$} In this section we show that $(\mathcal{M}_f^\infty,\OT_\infty)$ is complete. We will in fact prove that $\mathcal{M}_f^\infty$ is the Cauchy completion of $\mathcal{M}^\infty_{\textup{fin}}$, the subspace of finite measures. Since the finite measures are the ones that arise in practice, $\mathcal{M}_f^\infty$ is a convenient space to work with for both theoretical and practical reasons.


Recall that for every $\eps>0$ and $U \subset \w$, the \emph{$\epsilon$-thickening of $U$} is defined by $U^{\eps} := \{\mathbf{x} \in \w \ | \ \|\mathbf{x} - U\|_q < \eps\}$, where $\|\mathbf{x}-U\|_q := \inf_{\mathbf{u}\in U} \|\mathbf{x}-\mathbf{y}\|_q$.

To show that $\mathcal{M}_f^\infty$, we will show that it is a closed subspace of $\mathcal{M}^\infty$. For this, we need the following lemma which is reminiscent of an ``interleaving" theorem for measures. 

\begin{lemma}\label{lem:U_eps}
	Let $\eps>0$ and let $\mu, \nu \in \mathcal{M}^\infty$.  For any measurable $U \subset \w_{\eps}$, if $\OT_{\infty}(\mu, \nu) < \frac{\eps}{2}$ then $\mu(\overline{U}) \leq \nu(U^{\eps/2})$  and $\nu(\overline{U}) \leq \mu(U^{\eps/2})$.
	\begin{proof}
		Let $\mu, \nu \in \mathcal{M}^\infty$ such that $\OT_{\infty}(\mu, \nu) < \frac{\eps}{2}$. Let $\pi\in \Adm(\mu,\nu)$ with $C_\infty(\pi)<\frac{\eps}{2}$. We claim that $\pi(\overline{U}\times \ow) = \pi(\overline{U} \times U^{\eps/2})$. Indeed, if $(\mathbf{x},\mathbf{x}')\in \spt(\pi)$ with $\mathbf{x}\in \overline{U} \subset \ow_\eps$ then $d(\mathbf{x},\mathbf{x}') \leq C_\infty(\pi) < \frac{\eps}{2}$ so that $\mathbf{x}'\in U^{\eps/2}$. Thus $\overline{U} \times (\ow \backslash U^{\eps/2}) \subset (\ow\times \ow)\backslash \spt(\pi)$ so that, by Lemma \ref{lem:cplmt_spt_null}, $\pi(\overline{U} \times \ow) = \pi(\overline{U}\times U^{\eps/2}) + \pi(\overline{U} \times (\ow\backslash U^{\eps/2})) = \pi(\overline{U}\times U^{\eps/2})$, as claimed. Now $\mu(\overline{U}) = \pi(\overline{U}\times \ow) = \pi(\overline{U}\times U^{\eps/2}) \leq \pi(\ow \times U^{\eps/2}) = \nu(U^{\eps/2})$. Similarly, we can show that $\nu(\overline{U}) \leq \mu(U^{\eps/2})$.
	\end{proof}
\end{lemma}

Note that by setting $U=\ow_{\eps}$, $U^{\eps/2}=(\ow_{\eps})^{\eps/2}=\w_{\eps/2}$. Hence if $\OT_{\infty}(\mu, \nu) < \frac{\eps}{2}$ then $\mu(\ow_{\eps}) \leq \nu(\w_{\eps/2})$  and $\nu(\ow_{\eps}) \leq \mu(\w_{\varepsilon/2})$.

\begin{theorem}\label{thm:M_f_is_complete}
	$\mathcal{M}^\infty_f$ is closed in $(\mathcal{M}^\infty, \textup{OT}_\infty)$ and hence complete.
	\begin{proof}
		Let $\mu\in \mathcal{M}^\infty$ be a limit point of $\mathcal{M}^\infty_f$ and let $(\mu_n)\subset \mathcal{M}^\infty_f$ with $\OT_\infty(\mu_n,\mu)\to 0$. We claim that $\mu\in \mathcal{M}^\infty_f$. Fix $\varepsilon > 0$, choose $N\in \N$ large enough that $\OT_\infty(\mu,\mu_N) <\frac{\eps}{2}$. By Lemma \ref{lem:U_eps}, $\mu(\w_{\eps}) \leq \mu(\ow_{\eps}) \leq  \mu_N(\w_{\eps/2}) < \infty$ and thus $\mu \in \mathcal{M}_f^\infty$. By Theorem \ref{thm:completeness}, $(\mathcal{M}^\infty, \textup{OT}_\infty)$ is complete and hence so is $\mathcal{M}^\infty_f$.
	\end{proof}
\end{theorem}

Next, we show that $\mathcal{M}_f^\infty$ is the Cauchy completion of the subspace of finite measures. 

\begin{definition}\label{def:finite_measures} Let $\mathcal{M}^\infty_{\textup{fin}} := \{\mu\in \mathcal{M}^\infty \ | \ \mu(\mathbb{W}) <\infty\}$.
\end{definition}
It is clear from the definition that $\mathcal{M}^\infty_{\textup{fin}}\subset \mathcal{M}^\infty_f$.

\begin{theorem}\label{thm:completion_finite} 
$\mathcal{M}_f^\infty$ is the Cauchy completion of $(\mathcal{M}^\infty_{\textup{fin}}, \OT_\infty)$.
	\begin{proof} Since $(\mathcal{M}^\infty,\OT_\infty)$ is complete, it suffices to show that $\mathcal{M}^\infty_{\textup{fin}}$ is dense in $\mathcal{M}^\infty_f$. Let $\eps > 0$ and $\mu\in \mathcal{M}^\infty_f$ be given. Let $\mu_\eps$ be the measure on $\mathbb{W}$ given by $\mu_\eps(E) := \mu(E\cap \mathbb{W}_\eps)$ for all $E\subset \w$ Borel. Then $\mu_\eps(\mathbb{W}) = \mu(\mathbb{W}_\eps)$ so that $\mu_\eps\in \mathcal{M}^\infty_{\textup{fin}}$. We will show that $\OT_\infty(\mu,\mu_\eps) \leq \eps$ by constructing a coupling $\pi\in \Adm(\mu,\mu_\eps)$ with $C_\infty(\pi) \leq\eps$.
	
	Let $p:\mathbb{W}\to \Delta$ be the orthogonal projection map and let $\Delta_\eps:\w\to \ow\times \ow$ be given by $\mathbf{x}\mapsto (\mathbf{x},p(\mathbf{x}))$ if $d(\mathbf{x},\Delta)\leq\eps$ and $\mathbf{x}\mapsto(\mathbf{x},\mathbf{x})$ otherwise. Then $\Delta_\eps$ is Borel measurable and hence we may define $\pi:= (\Delta_\eps)_*\mu\in \mathcal{M}(\ow\times \ow)$. To see that $\pi\in \Adm(\mu,\mu_\eps)$, note that for $E\subset \mathbb{W}$ Borel we have
	\[\pi(E\times \ow) = \mu(\Delta^{-1}_\eps(E\times \ow)) = \mu(E),\]
	and
	\[ \pi(\ow\times E) = \mu(\Delta^{-1}_\eps(\ow\times E)) = \mu(E\cap \w_\eps) = \mu_\eps(E).\]
	Next, we claim that
	\begin{equation}\label{eq:spt_pi}
	\spt(\pi)\subset \{(\mathbf{x},\mathbf{y}) \subset \ow\times \ow \ | \ d(\mathbf{x},\mathbf{y}) \leq\eps\}.
	\end{equation}
	Indeed, let $(\mathbf{x},\mathbf{y})\in \ow\times \ow$ with $d(\mathbf{x},\mathbf{y})> \eps$. We will show that $(\mathbf{x},\mathbf{y})$ admits of neighborhood whose measure under $\pi$ is zero. First, note that it cannot be the case that both $\mathbf{x}, \mathbf{y}$ lie in $\Delta$ (or else we would have $d(\mathbf{x},\mathbf{y}) = 0$). Suppose that $\mathbf{x}\in\Delta$. Then we may choose open neighborhoods $U,V$ of $\Delta,\mathbf{y}$, respectively, with $d(U,V) > \eps$. Then $\Delta^{-1}_\eps(U\times V) = \emptyset$. Indeed, if $\mathbf{z}\in \Delta^{-1}_\eps(U\times V)$ then it must be the case that $\mathbf{z}\in U$. But then either $\mathbf{z} \in V$ or $d(\mathbf{z},\Delta) \leq\eps$, both of which are impossible since $d(U,V) >\eps$. Thus $\pi(U\times V) = \mu(\Delta^{-1}_\eps(U\times V)) = \mu(\emptyset) = 0$. A similar argument applies in the case that $\mathbf{y}\in \Delta$. Finally, suppose that neither of $\mathbf{x},\mathbf{y}$ lie in $\Delta$. In this case, we may choose open neighborhoods $U,V$ of $\mathbf{x},\mathbf{y}$, respectively, neither of which intersect $\Delta$ and with $d(U,V) > \eps$. But then $\Delta^{-1}_\eps(U\times V)$ is empty, since if $\mathbf{z}\in \Delta^{-1}_\eps(U\times V)$ then $\mathbf{z}\in U\cap V$, contradicting $d(U,V) > \eps$. Thus, in any case, we see that $(\mathbf{x},\mathbf{y}) \not\in \spt(\pi)$ whenever $d(\mathbf{x},\mathbf{y})> \eps$, which proves the claim.
	
	It now follows immediately from \eqref{eq:spt_pi} that $C_\infty(\pi) \leq \eps$ and thus $\OT_\infty(\mu,\mu_\eps)\leq \eps$, completing the proof.
\end{proof}
\end{theorem}

\subsection{Necessary Conditions for Relative Compactness in $\mathcal{M}_f^\infty$}
In this section we give necessary conditions for relative compactness in $\mathcal{M}_f^\infty$. These conditions are analogous to the necessary and sufficient conditions for relative compactness of subsets of persistence diagrams given in \cite{perea2022approximating}. However, in the measure-theoretic setting, these conditions fail to be sufficient. We end the section by providing an example to demonstrate why a full characterization of relative compactness is difficult (Lemma \ref{lem:OT_inf_dist_between_diracs}).

	Recall that a subset $S \subset \mathcal{M}^\infty_f$ is bounded if there is $\mu \in S$ and $M >0$ such that for all $\nu \in S$ we have $\OT_{\infty}(\mu, \nu) < M$. This gives our first necessary condition for compactness.

\begin{lemma}\label{lem:boundedness}
	If $S \subset \mathcal{M}^\infty_f$ is relatively compact then $S$ is bounded.
\end{lemma}

\begin{proof}
	Since $S \subset \mathcal{M}^\infty_f$ is relatively compact, $\overline{S}$ is compact. 
	Hence $\overline{S}$ is totally bounded which implies that $\overline{S}$ is bounded.
\end{proof}

Elements of $\mathcal{M}_f^\infty$ are by definition finite on all sets of the form $\ow_\eps$ for $\eps > 0$. Our next condition gives us uniform control of masses of the $\ow_\eps$ over all elements of a subset of measures.

\begin{definition}\label{def:UODF}
	A subset $S \subset \mathcal{M}_f^\infty$ is uniformly off-diagonally finite (UODF) if for every $\eps>0$, $\sup\{\mu(\ow_{\eps}) \ | \ \mu \in S\}  < \infty$.
\end{definition}

\begin{lemma}\label{lem:UODF}
	If $S \subset \mathcal{M}_f^\infty$ is relatively compact then $S$ is uniformly off-diagonally finite.
\end{lemma}

\begin{proof}
	Suppose $S \subset \mathcal{M}_f^\infty$ is relatively compact and not UODF. Then there exist $\eps > 0$ and a sequence $(\mu_n)$ in $S$ such that $\lfloor \mu_n(\ow_{\eps}) \rfloor < \lfloor \mu_{n+1}(\ow_{\eps})\rfloor $ for all $n$. Since $S$ is relatively compact, there is a subsequence $(\mu_{n_k})$ of $(\mu_n)$ such that $(\mu_{n_k}) \to \mu$ in $\mathcal{M}_f^\infty$. Let $L=\mu(\w_{\eps/2})$. Since $\mu \in \mathcal{M}_f^\infty$, $L < \infty$. Then we can choose $K$ large enough that $\OT_{\infty}(\mu_{n_K}, \mu)< \frac{\eps}{2}$ and $\mu_{n_K}(\ow_{\eps}) > L$. However, by Lemma \ref{lem:U_eps},  $\mu_{n_K}(\ow_{\eps})  \leq \mu(\w_{\eps/2}) = L$ which leads us to a contradiction.
\end{proof}

Since elements of $\mathcal{M}_f^\infty$ are finite on each $\ow_\eps$, the restriction of an element $\mu\in \mathcal{M}_f^\infty$ to $\ow_\eps$ is \emph{tight}. That is, for all $\delta>0$, there exists a compact set $K\subset \ow_\eps$ such that $\mu(\ow\eps\backslash K)<\delta$. Our final condition gives us, for each $\eps>0$, uniform control on the compact set $K$.

\begin{definition}A subset $S \subset \mathcal{M}(X)$ is off-diagonally uniformly tight (ODUT) if for all $\eps > 0$ and $\delta > 0$, there exists $N\in \N$ such that
	$\sup\{\mu(\w_\eps \cap(\R\backslash[-N,N] \times \R) \ | \ \mu\in S\} < \delta$.	
\end{definition}

\begin{lemma} \label{lem:ODUT} If $S\subset \mathcal{M}^\infty_f$ is relatively compact then $S$ is off-diagonally uniformly tight.
\end{lemma}
\begin{proof}
	Suppose $S \subset \mathcal{M}_f^\infty$ is relatively compact but not ODUT. Then $S$ is bounded, so there exists $\eps>0$ and $\delta>0$ such that, for each $n\geq 1$, there exists $\mu_n \in S$ with $\mu_n(U_n) \geq \delta$, where $U_n=\w_{\eps} \cap (\R\backslash[-n,n] \times \R)$. Since $S$ is relatively compact,  there exists a subsequence $(\mu_{n_k})$ of $(\mu_n)$ and $\mu\in \mathcal{M}_f^\infty$ such that $\OT_\infty(\mu_{n_k},\mu) \to 0$. Thus, there exists $K\in \N$ such that for all $k \geq K$  we have  $\OT_{\infty}(\mu_{n_k}, \mu) < \frac{\eps}{2}$. Note that for all $k$, $U_{n_{k+1}}^{\eps/2} \subset U_{n_{k}}^{\eps/2}$ and $\bigcap_{k=1}^{\infty} U_{n_k}^{\eps/2} =\emptyset$. By continuity of the measure $\mu$ from above, $\lim_{k\to \infty} \mu(U_{n_k}^{\eps/2})= \mu(\bigcap_{k=1}^{\infty} U_{n_k}^{\eps/2} )=0$. However, by Lemma \ref{lem:U_eps}, for all $k \geq K$, we have $\delta \leq \mu_{n_k}(U_{n_k}) \leq  \mu_{n_k}(\overline{U}_{n_k}) \leq \mu(U_{n_k}^{\eps/2})$ which leads us to a contradiction.
\end{proof}

We have thus shown that if $S \subset \mathcal{M}^\infty_f$ is relatively compact then $S$ is bounded, UODF, and ODUT. The next example shows that these conditions are not sufficient.

\begin{counterexample}
	Let $\mathbf{x}\in \w$ and let $S = \{\mu_n:=\frac{1}{n}\delta_x \ | \ n\in \N\}$. Then $S\subset \mathcal{M}_f^\infty$, is bounded (since $\OT_\infty(\mu_n,0) = \pers_\infty(\mu_n) = d(\mathbf{x},\Delta)$ for all $\mu_n$), is UODF (indeed, $S$ is uniformly finite), and is ODUT (for any $\eps > 0$, any compact set $K\subset \w$ containing $\mathbf{x}$ will satisfy the definition of ODUT). Now suppose that $S$ is relatively compact. Then $(\mu_n)$ admits a subsequence $(\mu_{n_k})$ converging to some $\mu\in \mathcal{M}_f^\infty$ in the $\OT_\infty$ distance. By Proposition \ref{prop:OT_conv_implies_vague_conv_and_pers_conv}, $\mu_{n_k}$ converges vaguely to $\mu$. But $\mu_{n_k}(f) = \frac{1}{n_k}f(x) \to 0$ as $k\to \infty$ for all $f\in C_c(\w)$. Thus $\mu(f) = 0$ for all $f\in C_c(\w)$ so that $\mu = 0$ by the Riesz-representation theorem \ref{thm:Riesz}. Hence $\OT_\infty(\mu_{n_k},0)\to 0$ as $k\to\infty$. But then $\OT_\infty(\mu_{n_k},0) = \pers_\infty(\mu_{n_k}) = d(\mathbf{x},\Delta)$ for all $k$, contradicting the fact that $\OT_\infty(\mu_{n_k},0)\to 0$. Thus $S$ is not relatively compact.
\end{counterexample}

The next lemma describes the property of the space $(\mathcal{M}^\infty_f, \OT_\infty)$ that makes the characterization of its relatively compact sets rather complicated. 

\begin{lemma}\label{lem:OT_inf_dist_between_diracs} Let $\mathbf{x}\in \w$ and let $0 \leq \alpha < \beta$. Then $\OT_\infty(\alpha\delta_{\mathbf{x}}, \beta\delta_{\mathbf{x}}) = d(\mathbf{x},\Delta)$.
		\begin{proof}
		Let $\pi = \alpha\delta_{(\mathbf{x},\mathbf{x})} + (\beta-\alpha)\delta_{(\mathbf{x},\Delta)}$. Then $\pi\in \Adm(\alpha\delta_{\mathbf{x}},\beta\delta_{\mathbf{x}})$ with $C_\infty(\pi) = d(\mathbf{x},\Delta)$ and hence $\OT_\infty(\alpha\delta_{\mathbf{x}},\beta\delta_{\mathbf{x}}) \leq d(\mathbf{x},\Delta)$.
		
		To obtain the reverse inequality, let $\pi\in \Opt_\infty(\alpha\delta_{\mathbf{x}},\beta\delta_{\mathbf{x}})$. By Remark \ref{rem:coupling_supoprted_off_diagonal}, we may assume that $\pi(\Delta\times \Delta) = 0$. We will show that there is a point of the form $(\mathbf{x},\Delta)$ or $(\Delta,\mathbf{x})$ in $\spt(\pi)$, from which the desired inequality will follow. Suppose this is not the case. Then $\pi(\w\times \Delta) = \pi(\Delta\times \w) = 0$. Now for any Borel set $E\subset \w$ containing $\mathbf{x}$, we have $\pi(E\times \ow) = \alpha$ and $\pi(\ow\times E) = \beta$. On the other hand, if a Borel set $E'\subset \w$ does not contain $\mathbf{x}$ then $\pi(E'\times \ow) = \pi(\ow\times E') = 0$. We claim that $\pi(\ow\times (\ow\backslash \{\mathbf{x}\})) = \pi((\ow\backslash \{\mathbf{x}\}) \times \ow) = 0$. Indeed, from the facts and assumptions above, we have
		\[ \pi(\ow\times (\ow\backslash\{\mathbf{x}\})) = \pi(\w\times (\w\backslash \{\mathbf{x}\})) + \pi(\w\times \Delta)) +  \pi(\Delta \times (\w\backslash \{\mathbf{x}\})) +  \pi(\Delta\times \Delta) = 0,\]
		and similarly for $\pi((\ow\backslash \{\mathbf{x}\}) \times \ow)$. Then
		\begin{multline*} \beta = \pi(\ow\times \{\mathbf{x}\}) = \pi(\ow\times \{\mathbf{x}\}) + \pi(\ow\times (\ow\backslash \{\mathbf{x}\}))\\
		= \pi(\ow\times \ow) = \pi(\{\mathbf{x}\}\times \ow) + \pi((\ow\backslash \{\mathbf{x}\})\times \ow) = \pi(\{\mathbf{x}\}\times \ow) = \alpha,
		\end{multline*}
		contradicting the fact that $\alpha < \beta$. Thus either $(\mathbf{x},\Delta)\in \spt(\pi)$ or $(\Delta,\mathbf{x})\in \spt(\pi)$, and the result follows.	
	\end{proof}
	
\end{lemma}

The preceding lemma implies non-separability of the spaces $(\mathcal{M}_{f}^\infty,\OT_\infty)$ and $(\mathcal{M}^\infty,\OT_\infty)$.

\begin{proposition}\label{prop:non_separable} $(\mathcal{M}_f^\infty,\OT_\infty)$ and $(\mathcal{M}^\infty,\OT_\infty)$ are not separable.
	\begin{proof}
		It suffices to find an uncountable family $S\subset \mathcal{M}_f^\infty$ and $\eps_0 > 0$ such that for all distinct $\mu,\nu\in S$, $\OT_\infty(\mu,\nu) \geq \eps_0$. Let $\mathbf{x} \in \w$ and take $S := \{\alpha \delta_{\mathbf{x}} \ | \ \alpha > 0\}$ and $\eps_0 :=  d(\mathbf{x},\Delta)$. Then by Lemma \ref{lem:OT_inf_dist_between_diracs}, we have $\OT_\infty(\mu,\nu) = d(\mathbf{x},\Delta) = \eps_0$ for all distinct $\mu,\nu$ in $S$ and the result follows.
	\end{proof}
\end{proposition}

Since (relatively) compact subsets of a metric space are separable, Proposition \ref{prop:non_separable} is a considerable obstruction to a complete characterization of the relatively compact subsets of $\mathcal{M}_f^\infty$.

\section{Experiments} We tested our learning scheme on four different datasets for classification tasks.

\subsection{Synthetic Shapes} As a proof-of-concept, we used our pipeline to classify compact manifolds (possibly with boundary) endowed with the uniform measure. Using the python package \texttt{teaspoon}\footnote{https://github.com/lizliz/teaspoon}, we computed uniform samples of the sphere, torus, circle, and annulus. Unsurprisingly, we were able to completely classify by shape computing persistence of the Vietoris-Rips complexes of large enough samples.

To highlight the benefit that our method has over computing persistence on a single sample, we first sampled 300 copies of each shape with a small sample size of 10. We then applied our featurization method to the resulting persistence diagram. For this, we used a step function on $[-0.1, 0.1]^2$ as a kernel to approximate an expected persistence measure $\rho$ of each mesh. For a template function $f$ we chose a step function supported on a square of length 0.4. Our template system $\{f_i\}_i$ consisted of 28 translations of $f$ so that their supports not overlap and cover the whole rectangular region that encloses all the persistence diagrams in the (birth, persistence) plane. We then applied multiclass polynomial logistic regression for classification.

We then repeated this experiment, but this time, for each of the 300 copies of the four shapes, we sampled 10 points 10 times. For each sample, we computed the persistence diagram and then used these 10 persistence diagrams to estimate the expected persistence diagram. We then applied our learning scheme with the same kernel and template described above together with polynomial logistic regression. We repeated this experiment with 20 samples per object and 40 samples per object. The classification accuracies are shown in Table \ref{tab:synthetic}.

\begin{table}
\begin{center}
	\begin{tabular}{|l | c | c | c| c|}
		\hline
		 Number of Samples per Object & 1 & 10 & 20 & 40\\
		\hline
		Classification Accuracy & 65.42\% & 94.58\% & 99.17\% & 100.00\%\\
		\hline
	\end{tabular}
\end{center}
\caption{The classification accuracies of our shape experiment. 400 instances of each of the sphere, torus, circle, and annulus were considered. We performed our experiment 4 times, sampling each shape 1, 10, 20, 40 times in each experiment. In each experiment, we computed the expected persistence diagram, computed features using our pipeline, and used polynomial logistic regression for classification.}\label{tab:synthetic}
\end{table}
Table \ref{tab:synthetic} shows that, expectedly, applying persistence to samples of just 10 points from each shape results in poor classification accuracy. However, by repeatedly sampling 10 points from each shape 40 times and computing the expected persistence measure, our method is able to achieve perfect classification. Moreover, repeatedly sampling a small number of points is much less computationally costly than computing the persistence for a large sample of points.

%
\subsection{Animal Poses}
Here we test our method for classification on 3D meshes of 6 animals each in 9 various poses \cite{sumner2004deformation}. From each mesh  20 times we sampled 1000 points with respect to normalized eccentricity using inverse transform sampling. Eccentricity was computed using the python package \texttt{GUDHI}\footnote{https://pypi.org/project/gudhi/}. For each sample,  we computed a geodesic distance matrix based on 50-nearest neighbors and then used it to computed a Vietoris-Rips persistence homology of degree 1. Then we used a step function on $[-0.0001, 0.0001]^2$ as a kernel to approximate an expected persistence measure $\rho$ of each mesh. For a template function $f$ we chose a step function supported on a square of length 0.02. Our template system $\{f_i\}_i$ consisted of 196  translations of $f$ so that their supports not overlap and cover the whole rectangular region that encloses all the persistence diagrams in the (birth, persistence) plane. Then for each mesh we computed its feature vector $(\int f_1 \hat{\rho} dx, \int f_2 \hat{\rho} dx \dots, \int f_{196} \hat{\rho} dx)$  where $\hat{\rho}$ is the kernel density estimation of $\rho$. By performing 3D multidimensional scaling on feature vectors, we were able to split animal meshes into 6 clusters which is shown in \cref{fig:mds_poses}. Finally, we applied multiclass polynomial regression to 80\% of feature vectors.  The remaining 20\% of feature vectors we used for testing and were able predict to which animal class they belong with 100\% accuracy.


\subsection{Texture Images}\label{sec:textures}
In this experiment we applied our method to classify gray-scale texture images from SIPI database. There are $13$ different texture classes, and within each class  images are taken under $7$ distinct rotation angles. For every image, we randomly sampled $100$ $35\times 35$ patches. For every patch, we computed a $0$-dimensional persistence diagram using a sublevelset filtration on pixel intensities. We estimated the density $\rho$ of the expected persistence measure of every texture image by kernel density estimation using a step function on the square $[-0.1, 0.1]^2$ as the kernel. For the template system $\{f_i\}_i$, we picked $340$ translations of a step function supported on a square of length $10$ and such that supports of $f_i$'s do not overlap but cover the whole rectangular region that encloses all the persistence diagrams in the (birth, persistence) plane. For every image, we computed the feature vector $(\int f_1 \hat{\rho} dx, \int f_2 \hat{\rho} dx \dots, \int f_{340} \hat{\rho} dx)$  where $\hat{\rho}$ is the kernel density estimation of $\rho$. For computation efficiency, we removed all zero columns from the feature matrix which as the result shortened feature vectors to length $209$. The $3$D-visualization of the multidimensional scaling of the feature vectors is shown in \cref{fig:mds_textures}. Next, we applied a polynomial regression using $3$rd degree polynomial to the feature vectors of $80\%$ of the data. Then we performed multiclass logistic regression to learn the coefficients of the polynomial that best approximates the labels of our training data. We used the remaining $20\%$ of the data to test our model and got accuracy equal to $100\%$.

\begin{figure}[t]
	\begin{center}
		\begin{tabular}{c c}
			\multicolumn{1}{l}{(a)} & \multicolumn{1}{l}{(b)}\\
			\includegraphics[trim={2cm -2cm 2cm 2.5cm},clip,width=0.6\linewidth]{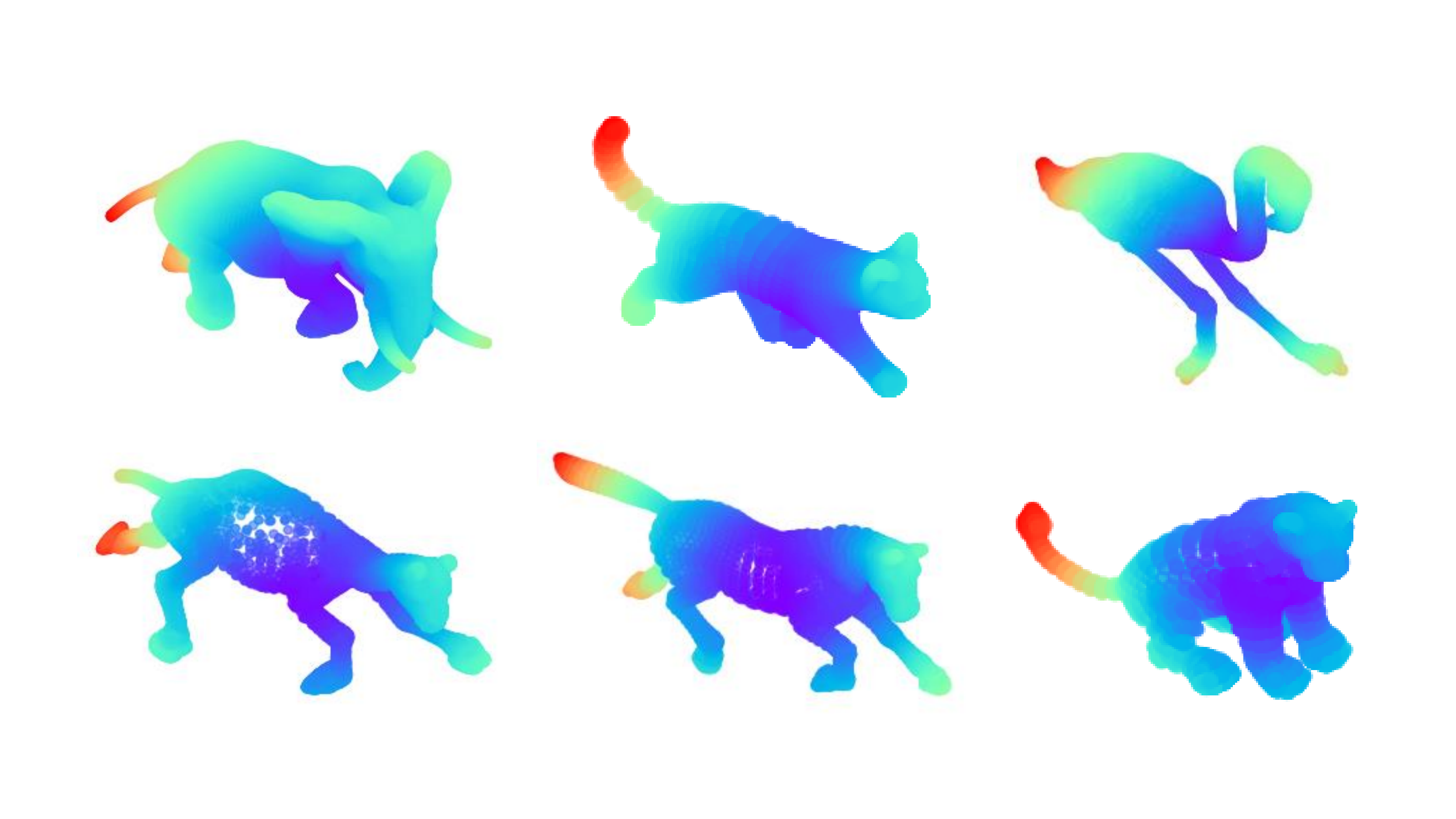}
			&\includegraphics[trim={4cm 0cm 0cm 2cm},clip,width=0.5\linewidth]{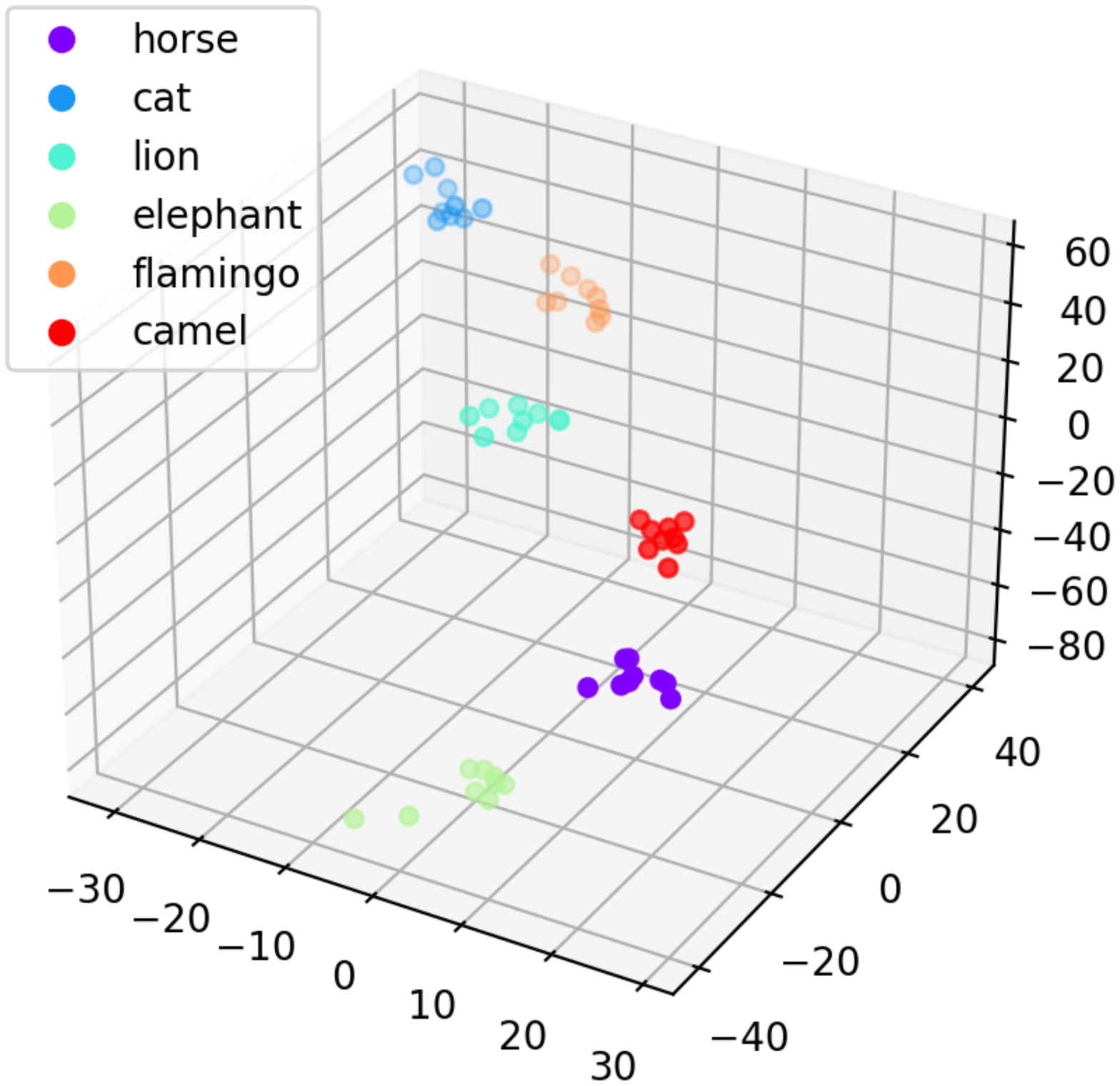} \\
		\end{tabular}
	\end{center}
	\caption{(a) Samples images of each of the 6 animal poses. Color corresponds to normalized eccentricity. (b) Multidimensional scaling of features feature vectors of each animal pose.}
	\label{fig:mds_textures}
\end{figure}

\begin{figure}[t]
	\begin{center}
		\begin{tabular}{c c}
			\multicolumn{1}{l}{(a)} & \multicolumn{1}{l}{(b)}\\
			\includegraphics[trim={6.5cm 3cm 6.5cm 3cm},clip,width=0.55\linewidth]{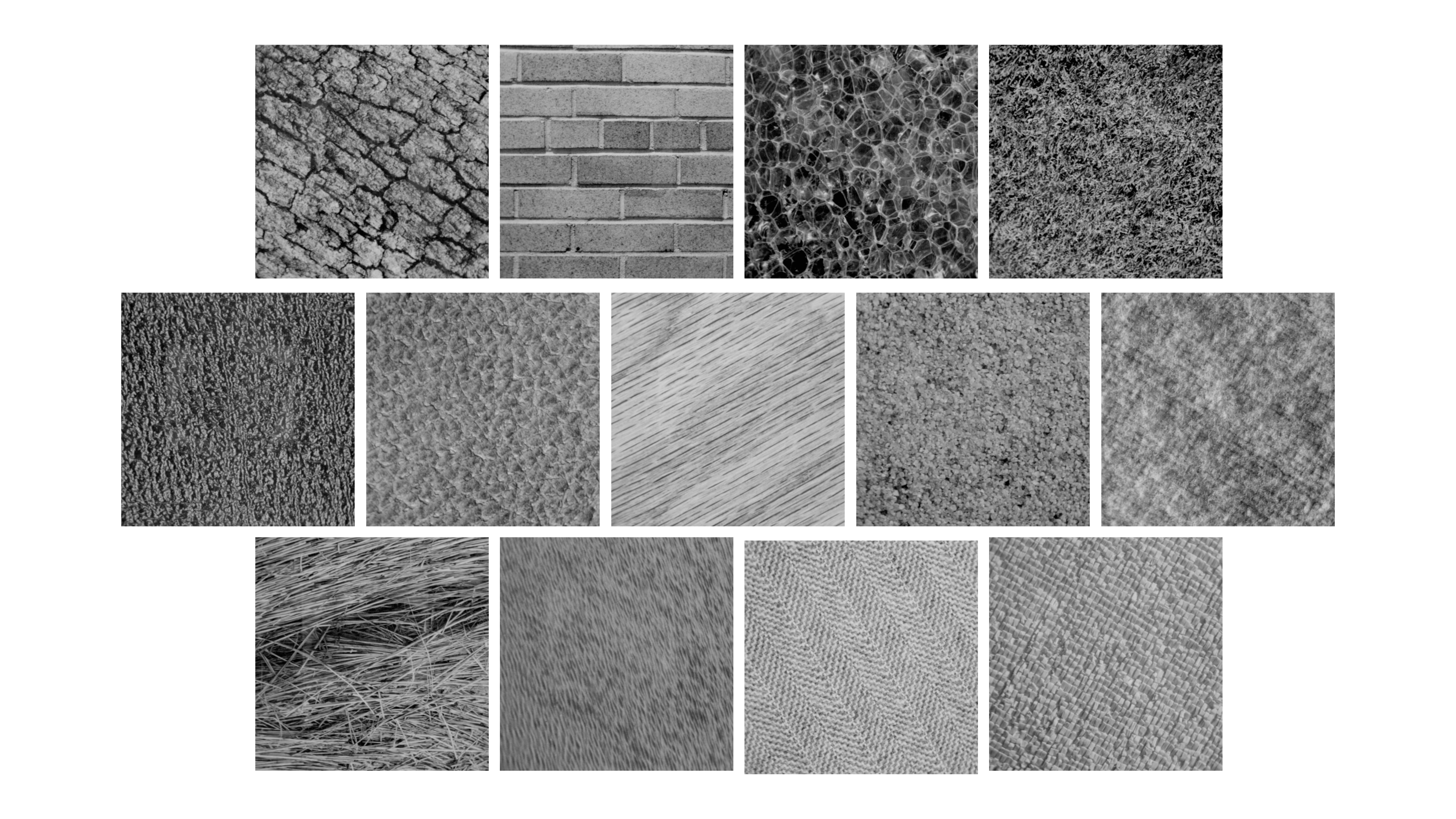}
			&\includegraphics[trim={0cm 0cm 0cm 0cm},clip,width=0.45\linewidth]{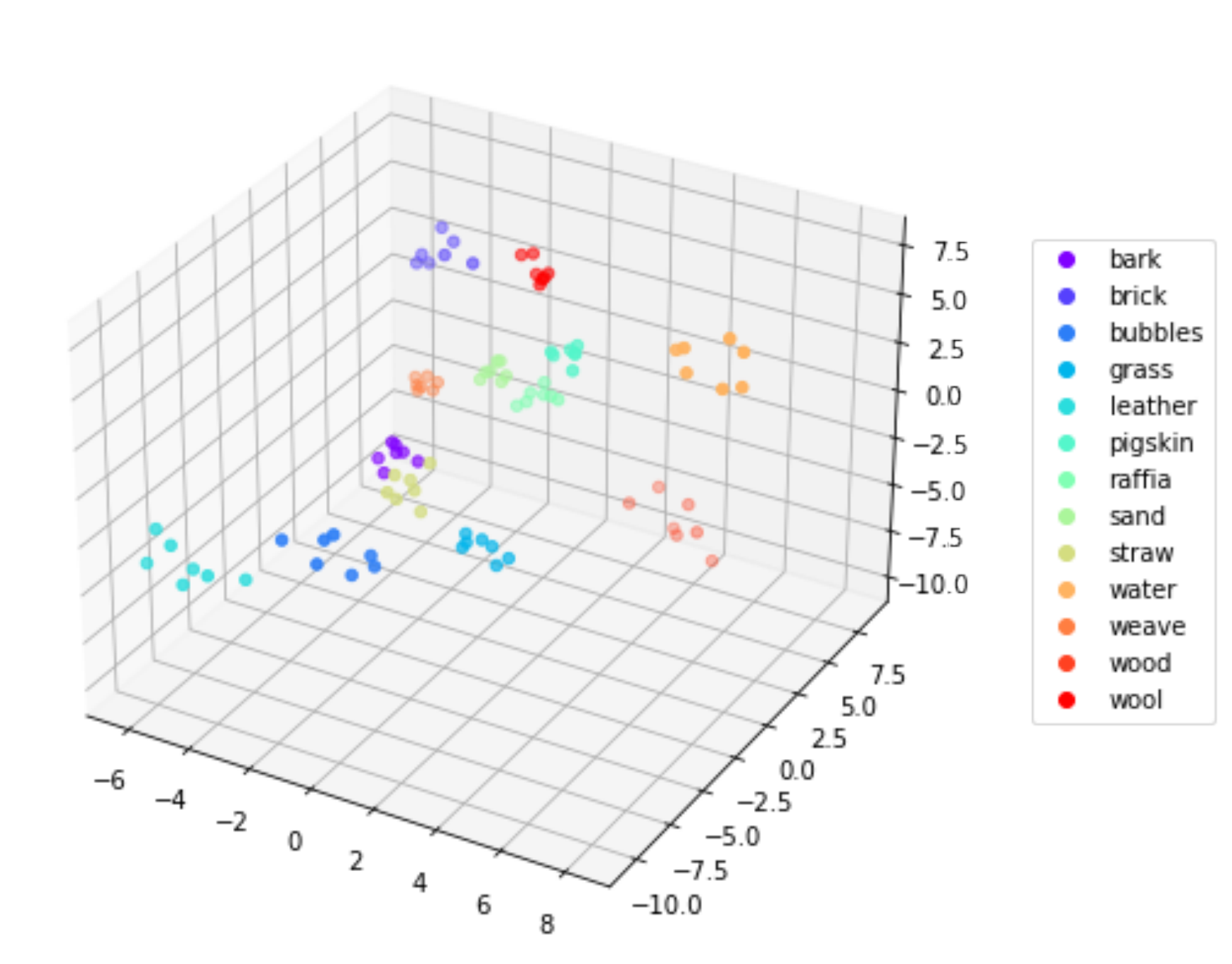} \\
		\end{tabular}
	\end{center}
	\caption{(a) Examples of each texture from the SIPI database. (b) Multidimensional scaling of the feature vectors of each texture image.}
	\label{fig:mds_poses}
\end{figure}

\subsection{Satellite Images of Cloud Patterns}

\begin{figure}[t]
	\begin{tabular}{c c c c}
			Sugar &Flower& Fish&Gravel \\
		\includegraphics[trim={0cm 0cm 0cm 0cm},clip,width=0.22\linewidth]{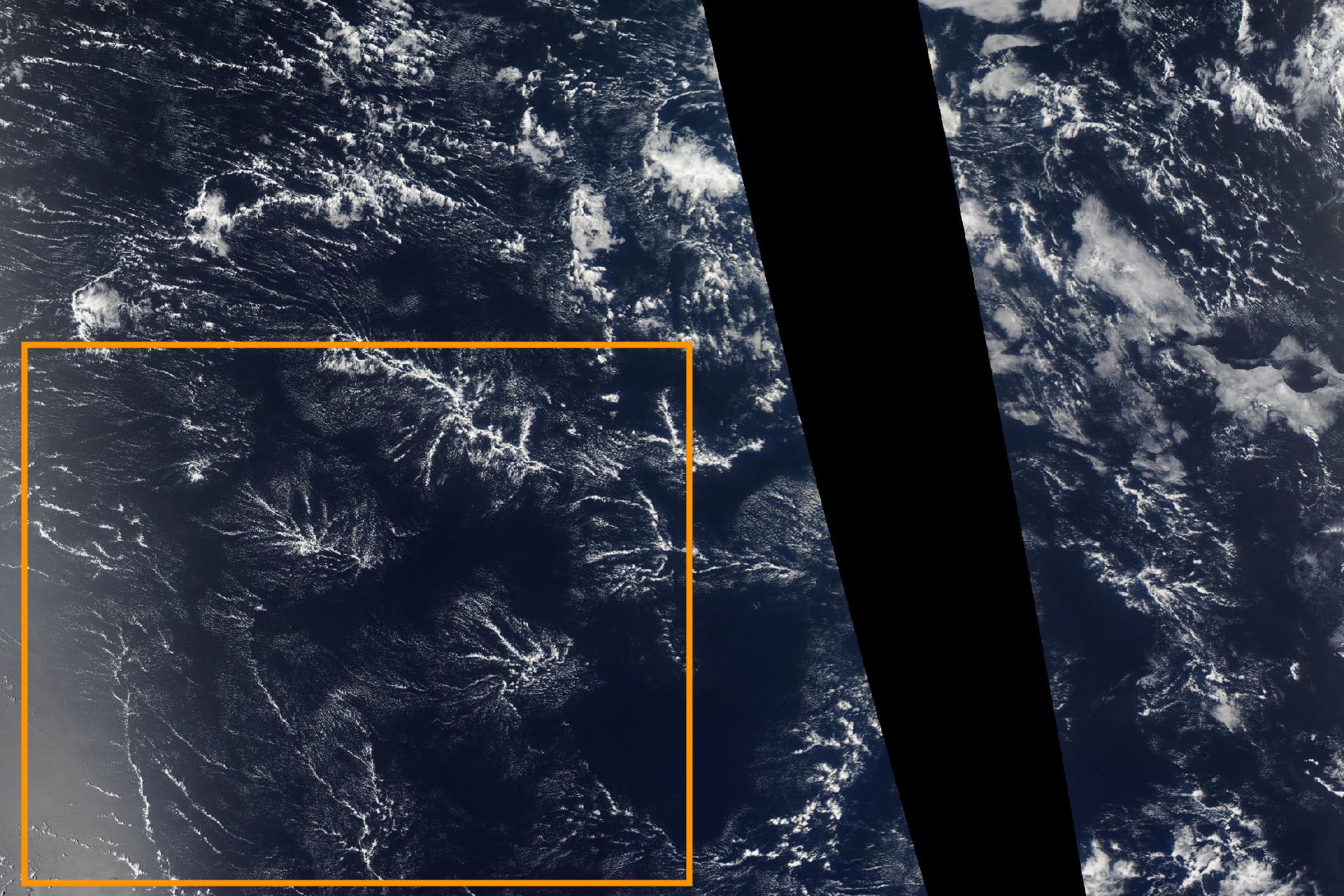} &
		\includegraphics[trim={0cm 0cm 0cm 0cm},clip,width=0.22\linewidth]{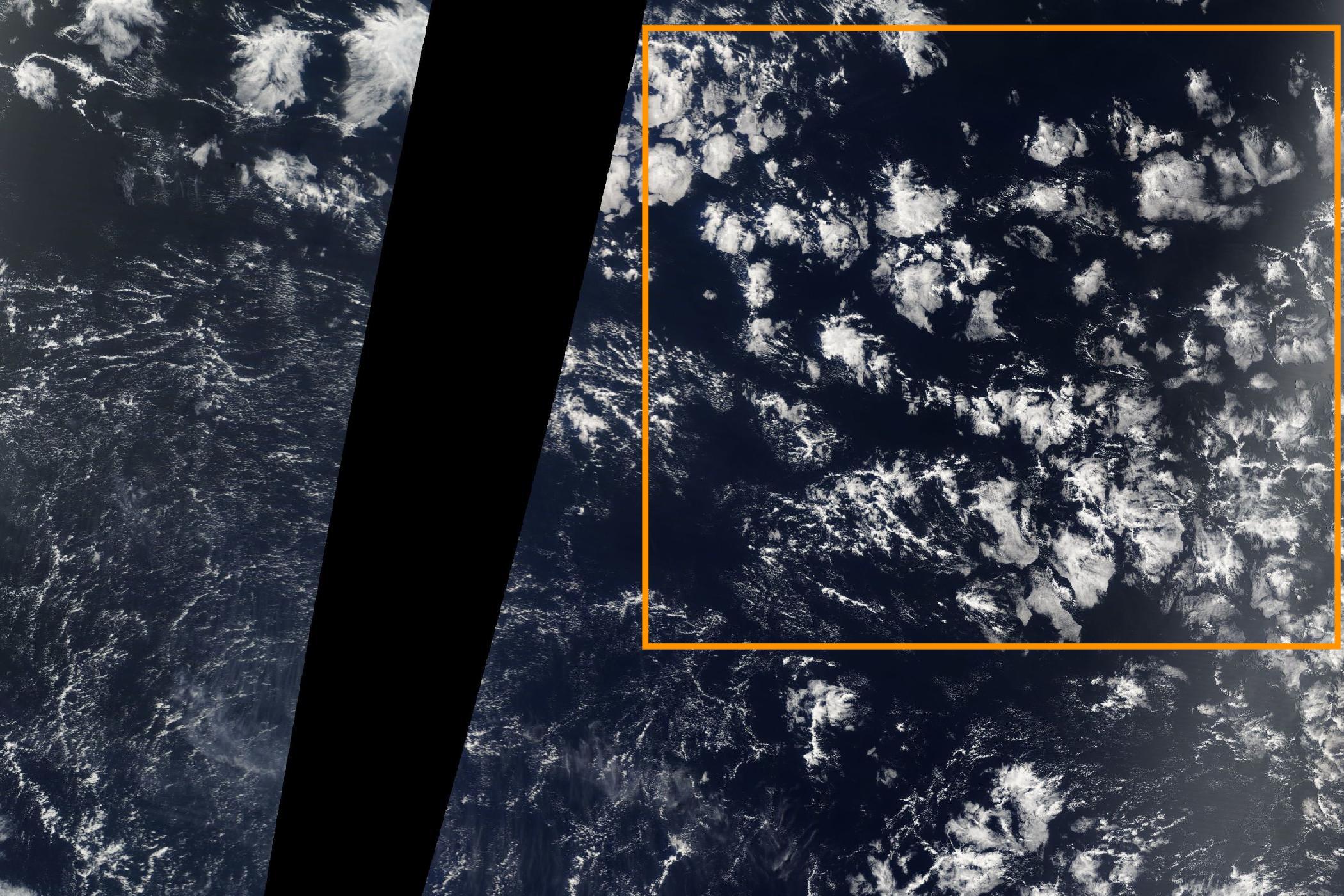} &
		\includegraphics[trim={0cm 0cm 0cm 0cm},clip,width=0.22\linewidth]{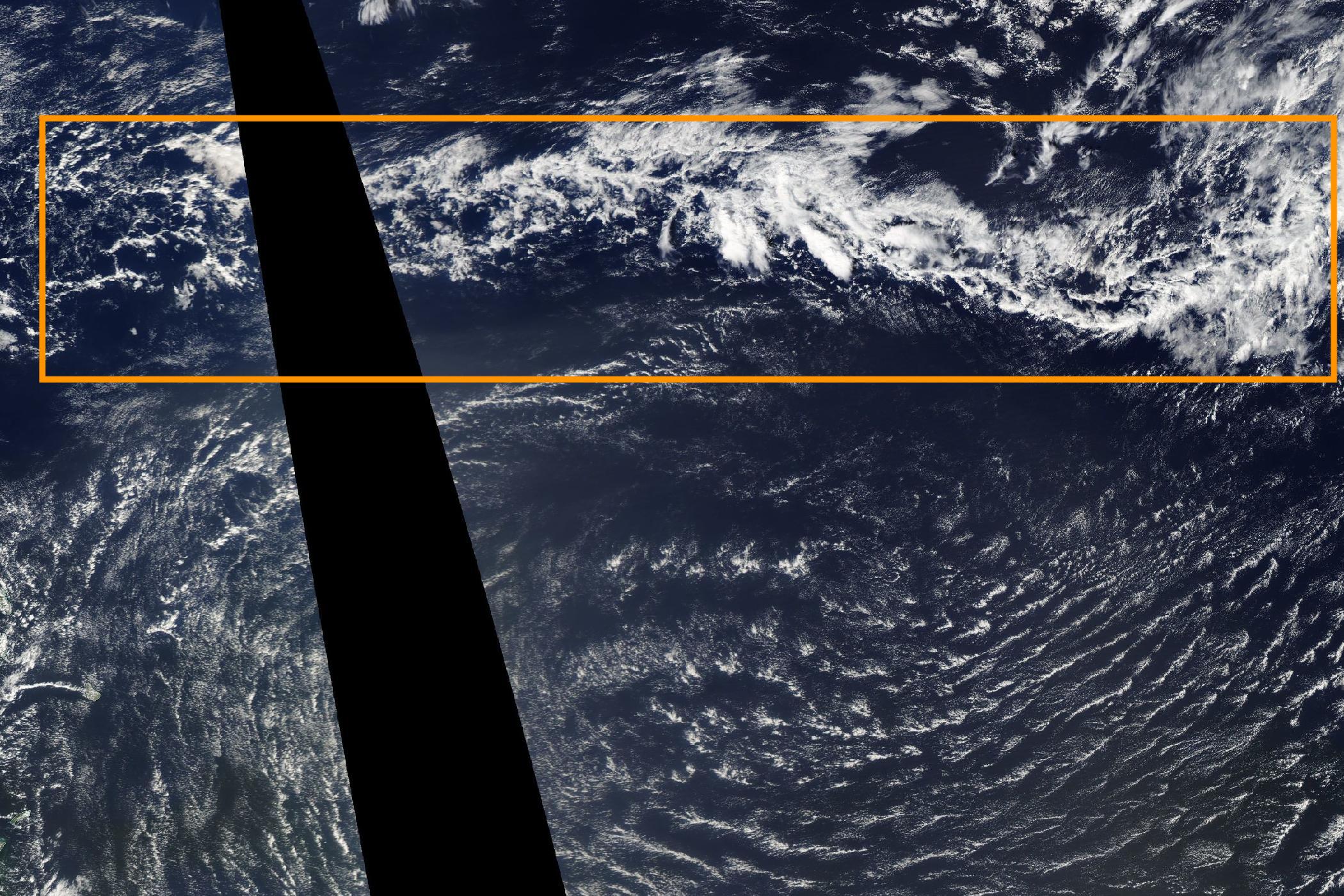} &
		\includegraphics[trim={0cm 0cm 0cm 0cm},clip,width=0.22\linewidth]{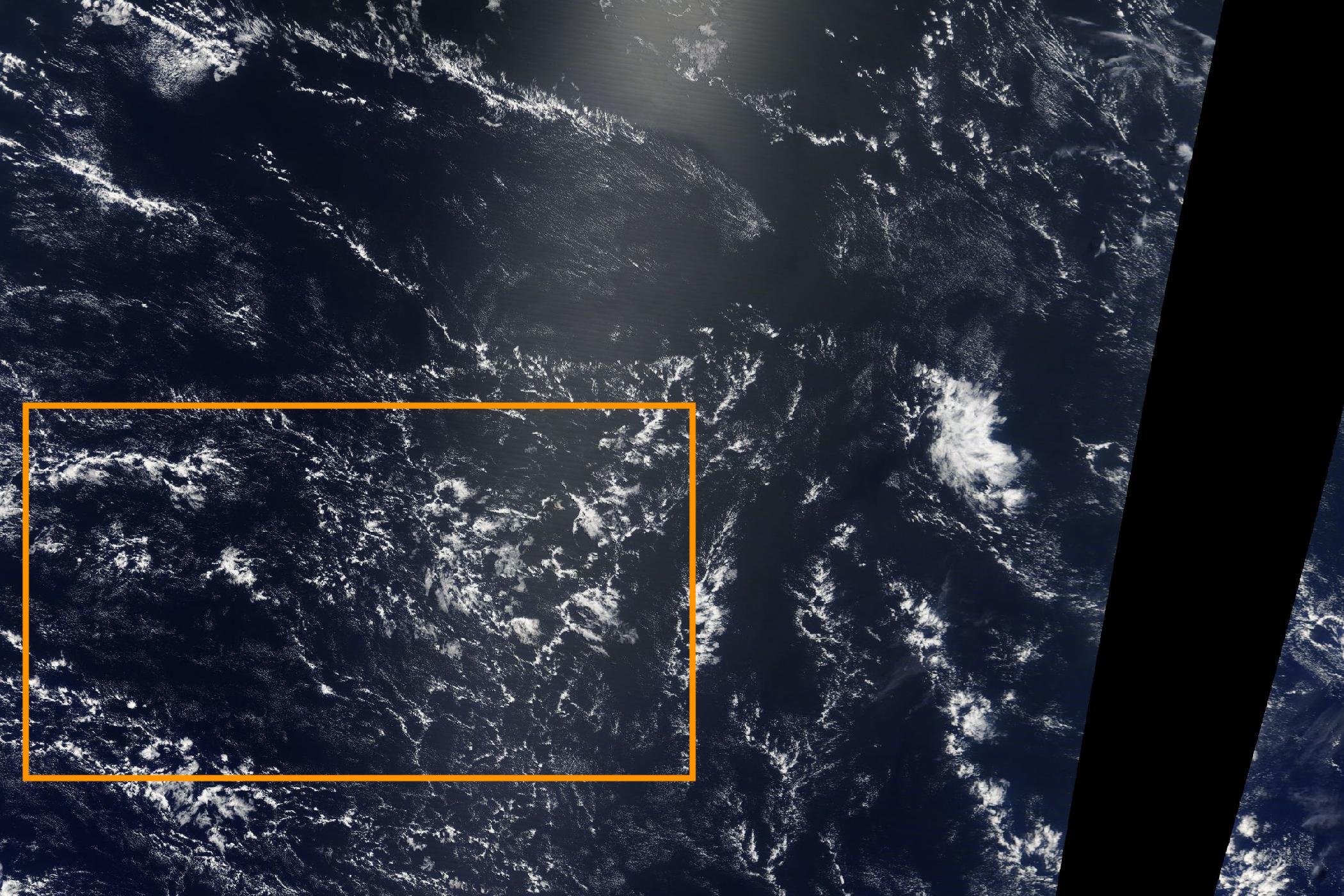}
\end{tabular}
		\caption{An example image from each of the 4 organization patterns with the annotations boxed in orange. These images are from the Aqua and Terra satellites and downlaoded from NASA Worldview using code from \cite{verhoef2022}. The labels and annotations are from \cite{rasp2020crowdsourcing}.}
		\label{fig:cloud_pics}
\end{figure}

Next, we applied our method to classifying cloud mesoscale organization patterns. This is an important task in environmental science, as it is known that understanding cloud organization patterns help predict future weather patterns \cite{stevens2020sugar}. In this example, we focus on 4 particular patterns: sugar, flower, fish, and gravel, as identified in \cite{stevens2020sugar}. We used data from \cite{rasp2020crowdsourcing}, where the authors used wide-scale crowdsourcing to reliably label 10,000 cloud images from the Aqua and Terra satellites with their respective clustering pattern. For each image, an expert in the field selected a rectangular region, called an \textit{annotation} and labeled that region with the correct pattern. In Figure \ref{fig:cloud_pics}, we see a sample image from each organization type and their respective annotation. In \cite{verhoef2022}, the authors show that topological data analysis is a good tool for classifying cloud images according to the 4 clustering types. Specifically, Ver Hoef et al. used persistence landscapes to train a support vector machine after projecting to 3 dimensions. In this section, we use our methods to train a logistic regression model on features arising from expected persistence diagrams from the 4 organization types.

We performed 3 binary classification tasks: gravel vs. flower, sugar vs. flower, and sugar vs. fish. For each organization pattern, we sampled 200 images and, from each image, we randomly sampled 50 $96 \times 96$ patches from the annotation. For each patch, we performed the same computations as in Section \ref{sec:textures}: that is, we computed persistence diagrams using sublevelset filtrations on pixel intensities and used the same kernel for kernel density estimation. In this example, our template system consisted of non-overlapping translations of a step function supported on a square of length 60, resulting in 12 translations $\{f_i\}_{i = 1}^{12}$ of our template function. Thus, if $\rho$ is the density of the expected persistence diagram of an image and $\hat{\rho}$ is its kernel density estimation, we computed a 12-dimensional feature vector $\left( \int f_1 \hat{\rho}dx, \int f_2 \hat{\rho}dx, \cdots \int f_{12} \hat{\rho}dx,\right)$ for each image. We trained a logistic regression model on the feature vectors from  80\% of the images and tested our model on the remaining 20\%. See Table \ref{tab:clouds} for how our method compares to that in \cite{verhoef2022}.

\begin{table}
\begin{center}
	\begin{tabular}{|l | c | c |}
		\hline
		Binary Classification & Accuracy in \cite{verhoef2022}& Our Accuracy \\
		\hline
		\hline
		Sugar vs. Fish & 86\% & 95\% \\
		\hline
		Sugar vs. Flower & 89.25\% & 93.75\% \\
		\hline
		Gravel vs. Flower & 81\% & 88.75\% \\
		\hline
	\end{tabular}
\end{center}
\caption{The classification accuracies of 3 binary classification tasks for cloud mesoscale organization patterns. We compare the results in \cite{verhoef2022}, where they used a support vector machine to classify the data, to our results, where we logistic regression. Notice that our method performs better in every task. }\label{tab:clouds}
\end{table}
\section{Discussion}
In this paper, we developed theoretical and computational methods for supervised learning on persistence measures. 
Specifically, we showed that the integration of Radon measures against re-scaled translates of any template $f \in C_c(\w)$ yield approximation results, in theory,
and expressive learning features in practice.
Persistence measures are thus a promising shape descriptor; in particular, for geometric objects of prohibitively large size --- so that a single persistent homology computation is unfeasible --- or in the presence of a measure or density function.
Several questions and research directions arise from this work. To name a few: can  statistical   bootstrap be leveraged to clarify  the relationship between the expected persistence measure, and  the persistence of the repeatedly sampled  metric measure space? is there a resulting probabilistic  stability theorem? can one fully characterize the compact subsets of $\mathcal{M}$ and $\mathcal{M}_f^\infty$ in terms of, perhaps, the  supports and moments of its elements?
We hope to pursue these and other questions in future work.

\subsection*{Acknowledgments}
We began this work at the American Mathematical Society's 2022 Mathematics Research Community ``Data Science at the Crossroads of Analysis, Geometry, and Topology", supported by the National Science Foundation under Grant Number DMS 1641020. J. A. Perea is partially supported by the National Science Foundation through grants CCF-2006661 and CAREER award DMS-1943758. The mesh data used in this project was made available by Robert Sumner and Jovan Popovic from the Computer Graphics Group at MIT. We would like to thank Lander Ver Hoef for providing code for preprocessing of the cloud satellite images dataset. We would like to thank Nate Mankovich and Josiah Hounyo for valuable programming advice.

\bibliographystyle{alpha}
\bibliography{bib}

\end{document}